\documentclass[twocolumn]{IEEEtran}

\usepackage[ansinew]{inputenc}
\usepackage{color,placeins}
\usepackage{hyperref}
\usepackage{graphicx,amsmath}
\usepackage{graphics} 
\usepackage{epsfig,epstopdf} 
\usepackage{mathptmx} 
\usepackage{times} 
\usepackage{amsmath} 
\usepackage{amssymb}  
\usepackage{amsfonts}
\usepackage{amsmath}
\usepackage{amssymb}
\usepackage{psfrag}
\usepackage{float,cite}

\newtheorem{thm}{Theorem}[section]
\newtheorem{remark}{Remark}[section]

\newtheorem{ass}{Assumption}[section]

\def\bi{\begin{itemize}}
\def\ei{\end{itemize}}
\def\bn{\begin{enumerate}}
\def\en{\end{enumerate}}
\def\bq{\begin{eqnarray}}
\def\eq{\end{eqnarray}}
\def\be{\begin{equation}}
\def\ee{\end{equation}}
\def\bea{\begin{eqnarray}}
\def\eea{\end{eqnarray}}
\def\beann{\begin{eqnarray*}}
\def\eeann{\end{eqnarray*}}
\def\bsea{\begin{subeqnarray}}
\def\esea{\end{subeqnarray}}
\def\bmat{\left[ \begin{array}}
\def\emat{\end{array} \right]}
\newfont{\BB}{msbm10}
\newfont{\bb}{msbm8}

\newcommand{\bmx}{\begin{matrix}}
\newcommand{\emx}{\end{matrix}}
\newcommand{\ba}{\begin{array}}
\newcommand{\ea}{\end{array}}

\def\nn{\nonumber}
\def\bq{\begin{eqnarray}}
\def\eq{\end{eqnarray}}

\def\bsmat{\left[ \begin{smallmatrix}}
\def\esmat{\end{smallmatrix} \right]}

\title{State Estimation For An Agonistic-Antagonistic Muscle System$^\star$}

\author{Thang~Nguyen$^{1}$, Holly Warner$^{2}$, Hung~La$^{3}$, Hanieh Mohammadi$^{1}$, Dan Simon$^{1}$, and Hanz Richter$^{2}$ %
\thanks{$^{1}$ Department of Electrical Engineering and Computer Science, Cleveland State University,	Cleveland, Ohio 44115, USA}%
\thanks{$^{2}$ Department of Mechanical Engineering, Cleveland State University, Cleveland, Ohio 44115, USA}%
\thanks{$^{3}$ Advanced Robotics and Automation (ARA) Lab, Department of Computer Science and Engineering, University of Nevada, Reno, NV 89557, USA}%
\thanks{$^\star$ This  work  was  supported  by  National  Science  Foundation  grant 1544702.}
}

\begin{document}
\maketitle 

\begin{abstract}                
Research on assistive technology, rehabilitation, and prosthesis requires the understanding of human machine interaction, in which human muscular properties play a pivotal role. This paper studies a nonlinear agonistic-antagonistic muscle system based on the Hill muscle model. To investigate the characteristics of the muscle model, the problem of estimating the state variables and activation signals of the dual muscle system is considered. In this work, parameter uncertainty and unknown inputs are taken into account for the estimation problem. Three observers are presented: a high gain observer, a sliding mode observer, and an adaptive sliding mode observer. Theoretical analysis shows the convergence of the three observers. To facilitate numerical simulations, a backstepping controller is employed to drive the muscle system to track a desired trajectory. Numerical simulations reveal that the three observers are comparable and provide reliable estimates in noise free and noisy cases. The proposed schemes may serve as frameworks for estimation of complex multi-muscle systems, which could lead to intelligent exercise machines for adaptive training and rehabilitation, and adaptive prosthetics and exoskeletons.
\end{abstract}


\section{INTRODUCTION}
\noindent
The development of robotics research has facilitated studies on applications in assisting human in various scenarios, see \cite{Wang_etal2013,LeamanLa2017} and references therein. In \cite{Wang_etal2013}, improved functionality in persons with certain neurological disorders was addressed. In \cite{Lin_etal2016}, human-like mechanical impedance based on the simulation of the models of the human neuromuscular system was studied. In \cite{Xiong_etal2016}, several virtual agonist-antagonist muscle mechanisms were considered in control of multilegged animal walking, where the controller is  a combination of neural control with tunable muscle-like functions. In \cite{Na_etal2016}, the estimation of joint force using a biomechanical muscle model and peaks of surface electromyography was studied.

The design of prosthetic, orthotic, and functional neuromuscular stimulation systems requires the understanding of the coordination of the human body and the dynamical properties of muscles \cite{zajac1989muscle}. The intermuscular coordination can be studied based on classical models proposed by Hill, Wilkie, and Richie \cite{zajac1989muscle}. The most widely implemented model for simulating human muscles is the Hill model \cite{winters1990hill}. More complicated models, including partial differential equation \cite{huxley1957double} or finite element \cite{yucesoy2002three} models, have been introduced to capture the complex behavior of human muscles. For a balance between accuracy and computational realizability, the Hill muscle model is a prominent solution \cite{zajac1989muscle}.

Human muscles operate at many joints. For a given joint, muscles often act in pairs with one or more muscles on opposite sides. Each member of a pair is regarded as agonist or antagonist. In this paper, an agonistic-antagonistic muscle system based on the Hill muscle model is introduced to study coordination and estimate muscle parameters.  The agonistic-antagonistic muscle system is scalable in the sense that its dynamic behavior and characteristics can be extended to multi-joint, multi-muscle, and 3D systems. In \cite{Huang2017_TSMCS}, muscular activities of a dominant antagonistic muscle pair are employed to address a computationally efficient model of the arm endpoint stiffness behavior.

%
%
%
%
%

A variety of estimation problems for different muscle models have been addressed. In \cite{buchanan2004}, muscle forces, joint moments, and/or joint kinematics are estimated from electromyogram signals using forward dynamics. In \cite{Erdemir2007}, the estimation problem of individual muscle forces during human movement is solved using forward dynamics. In \cite{Mohammed2016}, the muscular torque is estimated using a nonlinear observer in a sliding mode controller of a human-driven knee joint orthosis. In \cite{Yamasaki2016}, the estimation of muscle activity is conducted using higher-order derivatives, static optimization, and forward-inverse dynamics. In \cite{Lin2010}, an inverse dynamic optimization problem is proposed to estimate muscle and contact forces in the knee during gait. In \cite{Zhao2017_TSMCS}, the trajectory tracking control problem of one-degree of freedom  manipulator system driven by a pneumatic artificial muscle is addressed, in which a novel extended state observer based on a generalized super-twisting algorithm is employed to deal with internal uncertainties and external disturbances.

There have been numerous estimation methods proposed to observe nonlinear systems, from high gain observers to sliding mode observers; see \cite{Atassi1999, Edwards2000,davilaTAC2005,Yan2007,Alwi2009, LEE_Automatica2015,LEE_SCL2016,Hou2017_TSMCS,He2016_TSMCS} and references therein. High gain observers can offer a high level of accuracy in estimating state variables and uncertainties \cite{LEE_Automatica2015,LEE_SCL2016,He2016_TSMCS}. Sliding mode observers exhibit similar performance in estimating state variables and unknown inputs \cite{Edwards2000,Yan2007,Alwi2009,Hou2017_TSMCS}. Therefore, sliding mode observers, which are based on sliding mode control, can be employed to address many problems in fault detection and isolation, in which important parameters such as state variables, faults or unknown inputs need to be reconstructed from the available information. While traditional sliding mode techniques require the knowledge of unknown inputs and uncertainties, recent adaptive sliding mode control methods have been developed to overcome this limit at the cost of complexity \cite{EDWARDS_Automatica2016,Edwards_IJC2016}.

Muscle systems are important in assistive technology, rehabilitation, and prosthesis related research, which involves human-machine interactions. In this paper, we aim to design a high gain observer, a conventional sliding mode observer, and a new adaptive sliding mode observer for our dual muscle system. The benefits of accurate state estimation for the agonistic-antagonistic muscle model offer useful frameworks to investigate several problems in human-machine interactions such as monitoring of human health state and gait analysis \cite{nguyen2015dynamic,juen2014health,azulay1996automatic}, 3-D human skeleton localization \cite{YuanChen2013}, human foot localization \cite{NguyenLa2017,NguyenLa_ACC2016}, artificial muscles \cite{Zhao2017_TSMCS}, etc.

The contribution of our research work lies in the construction and development of a high gain observer, a sliding mode observer, and an adaptive sliding mode observer for the agonistic-antagonistic muscle system where unknown inputs are taken into account. Our problem is more general than the works in \cite{He2016_TSMCS,Zhao2017_TSMCS}, in which unknown input estimation is not considered, and more general than \cite{Hou2017_TSMCS}, where modeling uncertainties are not taken into account. The high gain observer is designed based on recent results in \cite{LEE_Automatica2015,LEE_SCL2016}, which allows to estimate state variables and unknown inputs, from which activation signals are constructed. The conventional sliding mode observer is built based on the first order sliding mode and super-twisting algorithm developed in \cite{davilaTAC2005,edwards1998book}, for which bounds of unknown control inputs and uncertainty needs to be known. The third observer is developed based on recent results on dual layer adaptive sliding mode control \cite{EDWARDS_Automatica2016,Edwards_IJC2016}, which does not require knowledge of the bounds of unknown inputs and uncertainty.


The rest of the paper is organized as follows. Section \ref{DynamicModel} presents the problem formulation. Section \ref{Observer} introduces three observers to estimate state variables and activation signals. Section \ref{Example} shows numerical simulations to demonstrate the effectiveness of the proposed schemes, where Subsection  \ref{Control} presents a backstepping controller for the tracking control problem. Section \ref{Conclusions} concludes the paper.

\section{PROBLEM FORMULATION}\label{DynamicModel}
We study the agonistic-antagonistic muscle system where each muscle is based on the Hill muscle model \cite{zajac1989muscle}. The Hill muscle unit models several effects of the physical muscle. It is divided into two sections, the tendon and the muscle body. The tendon is modeled as a nonlinear stiffness that includes some amount of slack. Within the muscle body portion of the model, a nonlinear stiffness element, modeled similar to the tendon, and a force generation element are oriented in parallel. The tendon and muscle body components are then placed in series. The structure of the dual muscle system is described in Fig. \ref{syst0}, where the abbreviations $CE$, $SEE$, and $PE$ stand for the contractile, series elastic, and parallel elastic elements of the Hill muscle model. Because muscles can only apply force when contracting, two muscles are required to actuate the central mass $m$, which is a simple load selected for studying the fundamental dynamics of this system.

The lengths of the $CE$ and $SEE$ are denoted as $L_{Cj}$ and $L_{Sj}$ for muscle $j$ $(j=1,2)$, and the total length of the $j$th muscle is defined by
\be\label{Lmj}
	L_{mj}=L_{Cj}+L_{Sj}.
\ee
Let $L_{m1}$ be the position of the mass in Fig. \ref{syst0}, and the corresponding velocity is positive to the right.

\begin{figure}[t]
	\centering
	\includegraphics{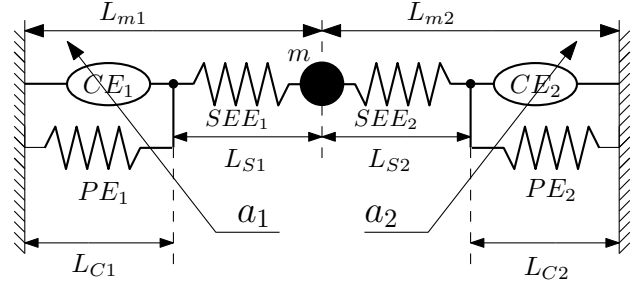}
	\caption{Two-muscle, one degree-of-freedom agonistic-antagonistic system with mass load  \cite{RichterWarnerIFAC17}.}
	\label{syst0}
\end{figure}

The dual muscle system possesses the following dynamics \cite{WarnerRichterTR2016,RichterWarnerIFAC17}
\bea
	\dot{x}_1 &=& x_2 \label{load1} \\
	\dot{x}_2 &=& \frac{1}{m}(\Phi_{S2}(L_{S2})-\Phi_{S1}(L_{S1})) +\Delta_{\Phi}(\tau) \label{load2} \\
	\dot{L}_{S1}&=& x_2+g_1^{-1}(z_1) \label{lsdot1} \\
	\dot{L}_{S2}&=& -x_2+g_2^{-1}(z_2) \label{lsdot2}
\eea
where
\bea
	x_1 &\triangleq &L_{m1},\\
	z_j&=&\frac{\Phi_{Sj}(L_{Sj})-\Phi_{Pj}(L_{Cj})}{a_jf_j(L_{Cj})} \text{ for } j=1,2, \label{z}
\eea
where $\Phi_{Sj}$ is the elastic force, $\Phi_{Pj}$ is the parallel elastic force, $a_j$ is the activation signal of element $j$ with $a_j\in[0,1]$, and $\Delta_{\Phi}(\tau)$ is a bounded uncertainty. The force-length dependence factor $f_j$ has the general shape of a Gaussian curve, and the velocity dependence function $g_j^{-1}(z_j)$ obeys the Hill model:
 \bea
 	f_j(L_{Cj})&=&exp[{-\left(\dfrac{L_{Cj}-1}{W}\right)^2}] \label{fj}\\
 	g_j^{-1}(z_j)&=&
 	\begin{cases}
 		\dfrac{1-z_j}{1+z_j/A},&z_j\leq1\\
 		\dfrac{-A(z_j-1)(g_{max}-1)}{(A+1)(g_{max}-z_j)},&z_j>1
 	\end{cases}\label{gj}
 \eea
where $W$, $A$, and $g_{max}$ are positive parameters.
Denote
\bea
	u_1&\triangleq &g_1^{-1}(z_1)\\
	u_2&\triangleq &g_2^{-1}(z_2)
\eea
as the virtual control inputs of the system (\ref{load1}), (\ref{load2}), (\ref{lsdot1}), (\ref{lsdot2}).

We have the following assumptions for our system.
\begin{ass}\label{As_Delta}
The uncertainty $\Delta_{\Phi}(\tau)$ satisfies
\be\label{Deltaphi}
	|\Delta_{\Phi}(\tau)|<\Delta_m
\ee
where $\Delta_m$ is a positive constant.
\end{ass}
\begin{remark}
$\Delta_{\Phi}(\tau)$ can represent parameter uncertainties due to model mismatch. For example, uncertainties in the description of $\Phi_{Sj}(L_{Sj})$ and the mass $m$.
\end{remark}
\begin{ass}\label{As_input}
The control inputs of the system (\ref{load1}), (\ref{load2}), (\ref{lsdot1}), (\ref{lsdot2}) satisfy
\be\label{Ujm}
	|u_j(\tau)|<U_{jm} \text{ for } j=1,2
\ee
where $U_{jm}$ is a positive constant.
\end{ass}
The length constraint of the dual muscle system is given by
\be\label{Lmi2}
	L_{m1}+L_{m2} =C
\ee
where $C$ is a constant. Hence, $L_{C1}$ and $L_{C2}$ will be determined from the relations in (\ref{Lmj}) and (\ref{Lmi2}) if $C$, $L_{S1}$, $L_{S2}$, and $L_{m1}$ are available. Therefore, it is sufficient to consider four differential equations of the model in (\ref{load1}), (\ref{load2}), (\ref{lsdot1}), and (\ref{lsdot2}) for our estimation problem. From (\ref{Lmj}), (\ref{lsdot1}), (\ref{lsdot2}), and  (\ref{Lmi2}), the dynamics of $L_{C1}$ and $L_{C2}$ are described as
\bea
	\dot{L}_{C1}&=&-g_1^{-1}(z_1) \label{lcdot1} \\
	\dot{L}_{C2}&=& -g_2^{-1}(z_2) \label{lcdot2}.
\eea

The nonlinear functions $\Phi_{Sj}$, $\Phi_{Pj}$, $f_j$, and $g_j^{-1}$ ($j=1,2$) can be found in \cite{WarnerRichterTR2016,RichterWarnerIFAC17}. All the variables and functions of the dual muscle system are normalized to simplify the dynamics. A candidate of $\Phi_{Sj}(L_{Sj})$ is chosen as \cite{WarnerRichterTR2016}
\small
\begin{align}
\nn&\Phi_{Sj}(L_{Sj})=\\
&\begin{cases}0& \text{$L_{Sj}<2$}\\
6760794.14(L_{Sj})^5-68434261.19(L_{Sj})^4\\
+277072371.99(L_{Sj})^3  -560875494.46(L_{Sj})^2\\
+567666340.97(L_{Sj})-229806913.40&\text{$2\leq L_{Sj}<2.04$}\\
 0.5 + 19.2308(L_{Sj}- 2.04)&\text{$L_{Sj}\geq 2.04$}\end{cases}\label{PhiSj}
\end{align}
\normalsize
whose graph is shown in Fig. \ref{Phi_fig}.
This function has the general shape of the tendon force-length characteristic, including slack. The piecewise polynomial in the expression of $\Phi_{Sj}(L_{Sj})$ is continuous up to the second derivative. An example of $\Phi_{Pj}$ is given as \cite{WarnerRichterTR2016}
\be
	\Phi_{Pj}(L_{Cj})=
 	\begin{cases}
 		0,&L_{Cj}<1\\
 		8\,(L_{Cj})^3-24\,(L_{Cj})^2+24\,L_{Cj}-8,&L_{Cj}\geq1.
 	\end{cases}\label{PhiPj}
\ee
\begin{remark}
The function $\Phi_{Sj}(L_{Sj})$ in (\ref{PhiSj}) is just one possibility to capture the stress-strain curve of a tendon. The shape of $\Phi_{Sj}(L_{Sj})$ can be built from data extracted from experiments. Note that the exact shape of $\Phi_{Sj}(L_{Sj})$ is not important as long as this function is known to controllers and observers.
\end{remark}

\begin{figure}[t]
	\centering
	\includegraphics[width=\linewidth]{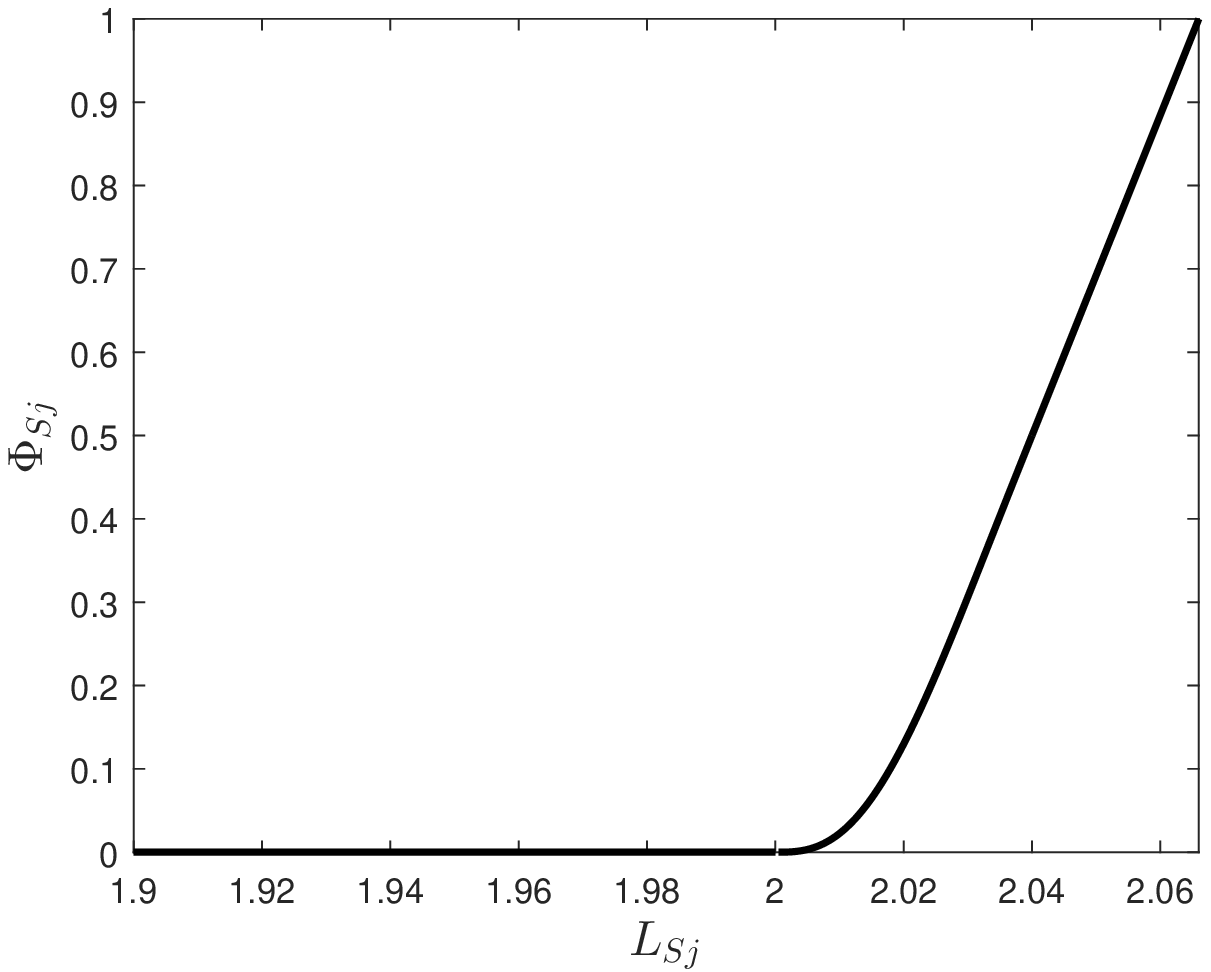}
	\caption{The graph of function $\Phi_{S\lowercase{j}}$ described in (\ref{PhiSj}).}
	\label{Phi_fig}
\end{figure}

Assume that $x_1$, $\Phi_{S1}(L_{S1})$, and $\Phi_{S2}(L_{S2})$ are available for measurement. The mass position can be tracked by a sensor while the $SEE$ nonlinear spring forces $\Phi_{S1}(L_{S1})$ and $\Phi_{S2}(L_{S2})$ of the agonistic-antagonistic muscles can be measured by two load cells, from which $L_{Sj}$ is inferred due to the inverse  of $\Phi_{Sj}(L_{Sj})$.
The observability matrix of the dual muscle system can be calculated using the Lie derivatives of the outputs, and it has rank 4, implying that the dual muscle system is locally observable \cite{HermannKrener1977}.

For ease of presentation, let

\bea
	x_3&\triangleq &L_{S1}\\
	x_4&\triangleq &L_{S2}.
\eea

Due to the relations (\ref{Lmj}) and (\ref{Lmi2}), $L_{Cj}$ can be deduced from $L_{mj}$ and $L_{Sj}$. Our system is rewritten as
\bea
	\dot{x}_1 &=& x_2 \label{x1d} \\
	\dot{x}_2 &=& \frac{1}{m}(\Phi_{S2}(x_4)-\Phi_{S1}(x_3)) +\Delta_{\Phi}(\tau) \label{x2d} \\
	\dot{x}_3&=& x_2+u_1(x) \label{x3d} \\
	\dot{x}_4&=& -x_2+u_2(x) \label{x4d}\\
	y_1&=&x_1\label{y1}\\
	y_2&=&\Phi_{S1}(x_3)\label{y2}\\
	y_3&=&\Phi_{S2}(x_4)\label{y3}
\eea
where
\be
	x=\left[\begin{matrix}x_1\\x_2\\x_3\\x_4\end{matrix}\right],
\ee
and the vector
\be
	y=\left[\begin{matrix}y_1\\y_2\\y_3\end{matrix}\right]
\ee
is the output of the dual muscle system. Note that from the measurement of $y_2$ and $y_3$, $x_3$ and $x_4$ can be calculated due to the inverse of the function $\Phi_{Sj}(L_{Sj})$ in (\ref{PhiSj}).
Let
\be
	u=\left[\bmx u_1\\u_2\emx\right].
\ee
Given the measurements of the length of the agonistic muscle and muscle forces, we study the estimation problem of state and activation signals. Due to the relation (\ref{z}), it is sufficient to estimate the state and unknown inputs of the system (\ref{x1d}) - (\ref{y3}).

\section{OBSERVER DESIGN}\label{Observer}
In this section, we introduce three methods to estimate the state variables and the activation signals: a high gain observer, a sliding mode observer, and an adaptive sliding observer.
Denote the estimates of $x$, $u$, and $a=[a_1,a_2]^T$ as
\bea
\hat{x}&=&\left[\bmx\hat{x}_1\\\hat{x}_2\\\hat{x}_3\\\hat{x}_4\emx\right]\\
\hat{u}&=&\left[\bmx\hat{u}_1\\\hat{u}_2\emx\right]\\
\hat{a}&=&\left[\bmx\hat{a}_1\\\hat{a}_2\emx\right].
\eea

\subsection{HIGH GAIN OBSERVER}
\label{HGO}
The high gain observer in this subsection is designed based on the extended high gain observer approach reported in \cite{LEE_Automatica2015,LEE_SCL2016}. The structure of the proposed high gain observer is described as
\bea
 	\dot{\hat{x}}_1&=&\hat{x}_2 +\frac{h_{11}}{\epsilon_h}(y_1-\hat{x}_1)\label{HO1a}\\
\nn	\dot{\hat{x}}_2&=&\frac{1}{m}(y_3-y_2) +\hat{\Delta}_{\Phi}(t) \\
&&+\frac{h_{12}}{\epsilon^2_h}(y_1-\hat{x}_1)\label{HO1b}\\
	\dot{\hat{\Delta}}_\Phi&=&\frac{h_{13}}{\epsilon^3_h}(y_1-\hat{x}_1)\label{HO1c}\\
	\dot{\hat{x}}_3 &=& \hat{x}_2+\hat{u}_1+\frac{h_{21}}{\epsilon_h}(\Phi_{S1}^{-1}(y_2)-\hat{x}_3) \label{HO2a} \\
	\dot{\hat{u}}_1&=&\frac{h_{22}}{\epsilon^2_h}(\Phi_{S1}^{-1}(y_2)-\hat{x}_3) \label{HO2b} \\
	\dot{\hat{x}}_4 &=& -\hat{x}_2+\hat{u}_2+\frac{h_{31}}{\epsilon_h}(\Phi_{S2}^{-1}(y_3)-\hat{x}_4) \label{HO3a}\\
	\dot{\hat{u}}_2&=&\frac{h_{32}}{\epsilon^2_h}(\Phi_{S2}^{-1}(y_3)-\hat{x}_4), \label{HO3b}
\eea
where $\epsilon_h\in (0,1)$ is a design parameter, parameters $h_{11}$, $h_{12}$, $h_{13}$ are chosen such that the polynomial $s^3+h_{11}s^2+h_{12}s+h_{13}$ is Hurwitz, parameters $h_{ij}$ for $i=2,3$ and $j=1,2$ are chosen such that the polynomials
$s^2+h_{i1}s+h_{i2}$ are Hurwitz for $i=2,3$ \cite{LEE_Automatica2015}.
\begin{thm}\label{thm_ho}
Under Assumptions \ref{As_Delta} and \ref{As_input}, the state and input estimates of the high gain observer presented in (\ref{HO1a}) - (\ref{HO3b}) satisfy
\be
	\|x(\tau)-\hat{x}(\tau)\|\rightarrow 0
\ee
and
\be
	\|u_j(\tau)-\hat{u}_j(\tau)\|\rightarrow 0 \text{ for } j=1,2
\ee
as $\epsilon_h\rightarrow 0$ for $\tau\geq 0$.
\end{thm}
\begin{proof}
The proof is based on the construction of the extended high gain observer in \cite{LEE_Automatica2015,LEE_SCL2016}.
\end{proof}
\begin{remark}
Theorem \ref{thm_ho} states that if $\epsilon_h\rightarrow 0$, the state and unknown input estimates will be exactly the true values. Since $\epsilon_h\ne 0$, a practical choice of $\epsilon_h$ lies in the interval $(0,1)$.
\end{remark}
\begin{remark}
The proposed high gain observer requires the tuning of nine parameters: $\epsilon_h$, $h_{11}$, $h_{12}$, $h_{13}$, $h_{21}$, $h_{22}$, $h_{31}$, $h_{32}$.
\end{remark}

\subsection{SLIDING MODE OBSERVER}
\label{SMO}

Following the super twisting algorithm and the traditional sliding mode approach in \cite{davilaTAC2005,edwards1998book}, the sliding mode observer for our system possesses the following structure:
\bea
	\dot{\hat{x}}_1&=&\hat{x}_2 +v_{11}\label{SO1a}\\
	\dot{\hat{x}}_2&=&\frac{1}{m}(y_3-y_2)+v_{12}\label{SO1b}\\
	\dot{\hat{x}}_3 &=& \hat{x}_2+v_2 \label{SO2a}\\
	\dot{\hat{x}}_4 &=& -\hat{x}_2+v_3 \label{SO3a}
\eea
where
\bea
	v_{11}&=& \lambda_{11}\, |y_1-\hat{x}_1| ^{1/2}\mathrm{sign}(y_1-\hat{x}_1) \label{v11}\\
	v_{12}&=& \alpha_{11} \,\mathrm{sign}(y_1-\hat{x}_1)\label{v12}\\
	v_2&=& \alpha_2 \,\mathrm{sign}(\Phi_{S1}^{-1}(y_2)-\hat{x}_3)\label{v2}\\
	v_3&=& \alpha_3 \,\mathrm{sign}(\Phi_{S2}^{-1}(y_3)-\hat{x}_4)\label{v3}.
\eea
Here $\lambda_{11}$ and $\alpha_{11}$ are design parameters which can be chosen to satisfy the following inequalities \cite{davilaTAC2005}:
\bea
	\alpha_{11}&>&f^+ \label{alpha11} \\
	\lambda_{11} &>& \sqrt{\frac{2}{\alpha_{11}-f^+}}\frac{(\alpha_{11}+f^+)(1+p)}{(1-p)}\label{lambda11}
\eea
where $p$ is a positive constant such that $0<p<1$, $f^+>0$ is the upperbound of $\Delta_\Phi$: $|\Delta_\Phi|<f^+$. The parameters $\lambda_{11}$ and $\alpha_{11}$ can also be taken according to \cite{MorenoTAC2012}. The parameters $\alpha_2$ and $\alpha_3$ in (\ref{v2}) and (\ref{v3}) are chosen such that \cite{edwards1998book}
\bea
	U_{1m}&<&\alpha_2\label{alpha2}\\
	U_{2m}&<&\alpha_3\label{alpha3}
\eea
where $U_{1m}$ and $U_{2m}$ are defined in (\ref{Ujm}).
The reconstruction of the uncertainty $\Delta_{\Phi}$ and unknown inputs $u_1$ and $u_2$ is accomplished with low pass filters given as
\bea
	\tau_s \, \dot{\hat{\Delta}}_\Phi&=&-\hat{\Delta}_{\Phi}+v_{12} \label{SO1c}\\
	\tau_s\, \dot{\hat{u}}_1&=&-\hat{u}_1+v_2 \label{SO2b}\\
	\tau_s\, \dot{\hat{u}}_2&=&-\hat{u}_2+v_3 \label{SO3b}
\eea
where $\tau_s$ is a positive parameter.

We have the following result.
\begin{thm}\label{thm_so}
Under Assumptions \ref{As_Delta} and \ref{As_input}, there exists a positive number $\tau^\star$ such that the state and input estimates of the high gain observer presented in (\ref{SO1a}) - (\ref{SO3a}) and (\ref{SO1c}) - (\ref{SO3b}) satisfy
\be
	x(\tau)-\hat{x}(\tau)= 0
\ee
and
\be
	u(\tau)\rightarrow\hat{u}(\tau)
\ee
for $\tau\geq \tau^\star$.
\end{thm}
\begin{proof}
The proof follows the super-twisting algorithm and the standard sliding mode in \cite{davilaTAC2005,edwards1998book}. Let
\be
	e=x-\hat{x}.
\ee
The state estimation dynamics are
\bea
	\dot{e}_1&=&e_2-v_{11} \label{de1}\\
	\dot{e}_2&=&\Delta_\Phi(\tau)-v_{12} \label{de2}\\
	\dot{e}_3&=&e_2+u_1+v_2 \label{de3}\\
	\dot{e}_4&=&-e_2+u_2+v_3 \label{de4}.
\eea
According to \cite{davilaTAC2005}, there exists a number $\tau^{\star}_1>0$ such that $e_1(\tau)=0$ and $e_2(\tau)=0$ for $\tau\geq \tau^{\star}_1$. It is easy to show that $e_3$ and $e_4$ are bounded in the interval $[0,\tau^{\star}_1]$. Since the error dynamics of $e_3$ is the first order sliding mode for $\tau\ge \tau^{\star}_1$, there exists a number $\tau^{\star}_2\geq\tau^\star_1$ such that  $e_3(\tau)=0$ for $\tau\geq \tau^{\star}_2$ \cite{edwards1998book}. Using the same argument, there exists a number $\tau^\star_3\geq \tau^\star_1$ such that $e_4(\tau)=0$ for $\tau\geq \tau^{\star}_3$ . Therefore, $e=0$ for $\tau\geq \tau^\star=\max \{\tau^\star_1,\tau^\star_2,\tau^\star_3\}$.

According to \cite{davilaTAC2005,edwards1998book}, the injection signals $v_{12}$, $v_2$, and $v_3$ are employed to estimate $\Delta_\Phi$, $u_1$, and $u_2$ in (\ref{SO1c}), (\ref{SO2b}), (\ref{SO3b}), from which $\hat{\Delta}_{\Phi}\rightarrow \Delta_{\Phi}$ and $\hat{u}\rightarrow u$.
\end{proof}
\begin{remark}
A practical implementation of the sign function of the sliding mode observer is done using the following approximation:
\be\label{deltas}
	\mathrm{sign}(e) \approx \frac{e}{\delta_s +|e|},
\ee
which adds another design parameter for the observer, namely $\delta_s$.
\end{remark}

\begin{remark}
The proposed sliding mode observer requires the tuning of six parameters: $\lambda_{11}$, $\alpha_{11}$, $\alpha_2$, $\alpha_3$, $\tau_s$, $\delta_s$.
\end{remark}
\begin{remark}
The parameters of the sliding mode observer depend explicitly on the information of the bounds of the unknown inputs and uncertainty.
\end{remark}
\subsection{ADAPTIVE SLIDING MODE OBSERVER}
\label{ASMO}
The adaptive sliding mode observer for our system is designed based on the dual layer nested adaptive approaches in \cite{EDWARDS_Automatica2016,Edwards_IJC2016}. The proposed adaptive sliding mode observer is given as follows:
\bea
\nn	\dot{\hat{x}}_1&=&\hat{x}_2 +\alpha_a(\tau)\, |y_1-\hat{x}_1| ^{1/2}\mathrm{sign}(y_1-\hat{x}_1)\\
&&-\phi(y_1-\hat{x}_1,L_a)\label{AO1a}\\
	\dot{\hat{x}}_2&=&\beta_a(\tau)\mathrm{sign}(y_1-\hat{x}_1) \label{AO1b}\\
	\dot{\hat{\Delta}}_\Phi&=&\frac{1}{\tau_a }(-\hat{\Delta}_{\Phi}-\beta_a(\tau)\mathrm{sign}(y_1-\hat{x}_1)  \label{AO1c}\\
	\dot{\hat{x}}_3 &=& (k_1(\tau)+\eta_1)\mathrm{sign}(\Phi_{S1}^{-1}(y_2)-\hat{x}_3)\label{AO2a}\\
	\dot{\hat{u}}_1&=&\frac{1}{\tau_a }(-\hat{u}_1-(k_1(\tau)+\eta_1)\mathrm{sign}(\Phi_{S1}^{-1}(y_2)-\hat{x}_3)) \label{AO2b}\\
	\dot{\hat{x}}_4 &=& (k_2(\tau)+\eta_2)\mathrm{sign}(\Phi_{S2}^{-1}(y_3)-\hat{x}_4) \label{AO3a}\\
	\dot{\hat{u}}_2&=&\frac{1}{\tau_a }(-\hat{u}_2-(k_2(\tau)+\eta_2)\mathrm{sign}(\Phi_{S2}^{-1}(y_3)-\hat{x}_4)) \label{AO3b}
\eea
where $\tau_a$, $\eta_1$, and $\eta_2$ are positive design parameters,
\bea
	\alpha_a(\tau)&=&\sqrt{L_a(\tau)}\,\alpha_0\\
	\beta_a(\tau)&=&L_a(\tau)\,\beta_0,
\eea
where $\alpha_0$ and $\beta_0$ are fixed positive scalars and
\be
	\phi(e_1,L_a)=-\frac{\dot{L}_a(\tau)}{L_a(\tau)}\, e_1(\tau).
\ee
Define
\be\label{deltaa0}
	\delta_{a0}(\tau)=L_a(\tau)-\frac{1}{a\beta_0} \, |\hat{\Delta}_\Phi|-\epsilon_a
\ee
where $a$ is chosen such that $0<a<1/\beta_0<1$ and $\epsilon_a$ is a small positive scalar chosen to satisfy
\be
	\frac{1}{a\beta_0} \, |\hat{\Delta}_\Phi|+\epsilon_a/2> |\Delta_\Phi|.
\ee
The proposed adaptive element $L_a(\tau)$ is given by
\be\label{La}
	L_a(\tau)=l_{0}+l_a(\tau)
\ee
where $l_0$ is a small positive design constant and
\be\label{dotla}
	\dot{l}_a(\tau)=-\rho_{a0}(\tau)\mathrm{sign}(\delta_a(\tau)).
\ee
The time-varying term in (\ref{dotla}) is given by
\be\label{rhoa0}
	\rho_{a0}(\tau)=r_{00}+r_{a0}(\tau)
\ee
where $r_{00}$ is a positive design parameter,
\be\label{dra0}
	\dot{r}_{a0}(\tau)=\begin{cases} \gamma_{a0} \,|\delta_{a0}(\tau)| & \text{if } |\delta_{a0}(\tau)|>\delta_{00}\\
0&\text{otherwise}\end{cases}
\ee
where $\delta_{a0}$ is defined in (\ref{deltaa0}), $\gamma_{a0}>0$ and $\delta_{00}>0$ are design parameters.
For $j=1,2$, define
\be
	\delta_{aj}(\tau)=k_j(\tau)-\frac{1}{\alpha_{aj}} \, |\hat{u}_j|-\epsilon_{aj}
\ee
where $\alpha_{aj}$ is chosen such that $0<\alpha_{aj}<1$ and $\epsilon_{aj}>0$ is a small positive scalar chosen to satisfy
\be\label{uj}
	\frac{1}{\alpha_{aj}} \, |\hat{u}_j|+\epsilon_{aj}/2> |u_j|.
\ee
The proposed adaptive elements $k_j(\tau)$ are given by
\be\label{kj}
	\dot{k}_j(\tau)=-\rho_{aj}(\tau)\mathrm{sign}(\delta_{aj}(\tau))
\ee
for $j=1,2$.
The time-varying terms in (\ref{kj}) are given by
\be\label{rhoaj}
	\rho_{aj}(\tau)=r_{0j}+r_{aj}(\tau), \text{ for } j=1,2
\ee
where
\be\label{draj}
	\dot{r}_{aj}(\tau)=\begin{cases} \gamma_{aj} \,|\delta_{aj}(\tau)| & \text{if } |\delta_{aj}(\tau)|>\delta_{0j}\\
0&\text{otherwise}\end{cases}
\ee
where $\gamma_{aj}>0$ and $\delta_{0j}$ is a small positive parameter.

\begin{thm}\label{thm_ao}
Under Assumptions \ref{As_Delta} and \ref{As_input}, there exists a positive number $\tau^\dagger$ such that the state and input estimates of the high gain observer presented in (\ref{AO1a}) - (\ref{AO3b}) satisfy
\be
	x(\tau)-\hat{x}(\tau)= 0
\ee
and
\be
	u(\tau)\rightarrow \hat{u}(\tau)
\ee
for $\tau\geq \tau^\dagger$.
\end{thm}
\begin{proof}
The proof follows the results of the dual layer nested adaptive approaches in \cite{EDWARDS_Automatica2016,Edwards_IJC2016}. The error dynamics for the state estimation are
\bea
	\dot{e}_1&=&e_2 -\alpha_a(\tau)\, |e_1| ^{1/2}\mathrm{sign}(e_1)-\phi(e_1,L_a)\label{de1a}\\
	\dot{e}_2&=&\Delta_\Phi-\beta_a(\tau)\mathrm{sign}(e_1) \label{de2a}\\
	\dot{e}_3 &=& e_2-(k_1(\tau)+\eta_1)\mathrm{sign}(e_3)\label{de3a}\\
	\dot{e}_4 &=& -e_2-(k_2(\tau)+\eta_2)\mathrm{sign}(e_4) \label{de4a}.
\eea
According to \cite{Edwards_IJC2016}, there exists a number $\tau^\dagger_1>0$ such that $e_1(\tau)=0$ and $e_2(\tau)=0$ for $\tau\geq \tau^{\dagger}_1$. It is easy to show that $e_3$ and $e_4$ are bounded in the interval $[0,\tau^{\dagger}_1]$.

According to \cite{EDWARDS_Automatica2016}, there exists a number $\tau^{\dagger}_2\geq\tau^{\dagger}_1$ such that  $e_3(\tau)=0$ for $\tau\geq \tau^{\dagger}_2$ \cite{edwards1998book}. Using the same argument,  there exists a number $\tau^{\dagger}_3\geq \tau^{\dagger}_1$ such that  $e_4(\tau)=0$ for $\tau\geq \tau^{\dagger}_3$. Therefore, $e=0$ for $\tau\geq \tau^{\dagger} =\max \{\tau^{\dagger}_1,\tau^{\dagger}_2,\tau^{\dagger}_3\}$.

The recovery of $\Delta_\Phi$, $u_1$, and $u_2$ follows the standard filtering approach in sliding mode control \cite{edwards1998book} in (\ref{AO1c}), (\ref{AO2b}), (\ref{AO3b}), from which $\hat{\Delta}_{\Phi}\rightarrow \Delta_{\Phi}$ and $\hat{u}\rightarrow u$.
\end{proof}
\begin{remark}
Similar to the traditional sliding mode observer, the sign function of the adaptive sliding mode observer can be approximated using the expression in (\ref{deltas})
\be\label{deltaa}
	\mathrm{sign}(e) \approx \frac{e}{\delta_a +|e|},
\ee
which introduces another design parameter, that is $\delta_a$.
\end{remark}

\begin{remark}
The proposed adaptive sliding mode observer requires the tuning of 21 parameters: $\alpha_0$, $\beta_0$, $\eta_1$, $\eta_2$, $a$, $l_0$, $r_{00}$, $r_{01}$, $r_{02}$, $\tau_a$, $\epsilon_{a1}$, $\epsilon_{a2}$, $\alpha_{a1}$, $\alpha_{a2}$, $\gamma_{a0}$, $\gamma_{a1}$, $\gamma_{a2}$, $\delta_{00}$, $\delta_{01}$, $\delta_{02}$, $\delta_a$.
\end{remark}

\begin{remark}
The parameters of the adaptive sliding mode observer in general do not depend on the bounds of the unknown inputs and uncertainty.
\end{remark}

\section{NUMERICAL EXAMPLE}\label{Example}
For the purpose of estimation, we employ a backstepping controller for the output $L_{m1}$ to track a time-varying reference signal. A numerical example will be conducted using the proposed controller and observers to estimate the state variables and the activation signals.

\subsection{BACKSTEPPING CONTROLLER}\label{Control}
The specific controller is irrelevant for estimation analysis and design, as long as the estimates are not being fed back to the controller. This is the case even when the estimator does not have access to direct control input measurements, provided an accurate dynamic model is available.

In this paper, a tracking control scheme is constructed based on its counterpart for setpoint regulation \cite{RichterWarnerIFAC17}. A tracking extension for the dual muscle system, which includes activation dynamics, is reported in \cite{WarnerDSCC2017}. A control method based on an artificial field approach can be derived as in \cite{Woods2017_TSMCS}. Our goal is to design a stable feedback tracking controller for the position of the mass, in which $u_1$ and $u_2$ are control inputs. The activation signals $a_1$ and $a_2$ are subsequently calculated from the relation in (\ref{a}). Assume that the uncertainty $\Delta_\Phi$ is known to the controller.

As in \cite{RichterWarnerIFAC17}, the standard backstepping procedure is employed to synthesize a virtual control input based on tendon force difference to setpoint-stabilize the load subsystem formed by (\ref{load1}) and (\ref{load2}). The constructive scheme is based on a Lyapunov function $V$ that becomes negative-definite for the load subsystem under the synthetic control law.

In \cite{RichterWarnerIFAC17}, two alternative methods for the synthetic input are employed: a scalar approach and a vector approach. We aim to design our control method based on the former. Denote the reference signal as $r(t)$ and assume that it is twice differentiable.

Denote the tracking error and its derivative as
\be
	e=\left[\begin{matrix}
	e_1\\e_2
	\end{matrix}\right]=\left[\bmx x_1-r\\\dot{x}_1-\dot{r}\emx\right].
\ee
Furthermore, define
\be
	\zeta=\Phi_{S2}(L_{S2})-\Phi_{S1}(L_{S1})+m\Delta_\Phi-m\dot{r}.
\ee
Our goal is to design $u_1$ and $u_2$ such that $e$ converges to 0.
The error dynamics is described in the form
\be\label{edynamics}
	\dot{e}=Ae+B\zeta
\ee
where
\be
	A=\left[\begin{matrix}
	0&1\\0&0
	\end{matrix}\right], \quad B=\left[\begin{matrix}
	0\\\frac{1}{m}
	\end{matrix}\right].
\ee
Consider the Lyapunov function
\be
	V=\frac{1}{2}e^TPe
\ee
where $P$ is a positive definite matrix. The system (\ref{edynamics}) is stable if a state feedback regulator is chosen as $\zeta=\Psi(e)=-Ke$ such that $A_{cl}=A-BK$ is Hurwitz. Hence,
\be\label{dotV}
	\dot{V}=\frac{1}{2}e^T(A_{cl}^TP+PA_{cl})e=-\frac{1}{2}e^TQe
\ee
where $Q$ is positive definite. Thus, $\dot{V}<0$. This implies that the error converges to 0. However, $\zeta$ is not a direct control input. As a result, we introduce a variable
\be
	w=\zeta-\Psi(e).
\ee
Its derivative is given as
\be
	\dot{w}=\dot{\zeta}-\dot{\Psi}(e)=\Phi^\prime_{S2}\dot{L}_{S2}-\Phi^\prime_{S1}\dot{L}_{S1}+m\dot{\Delta}_\Phi-m\ddot{r}
\ee
where
\be
	\Phi^\prime_{Si}=\frac{d\Phi_{Si}}{dL_{Si}}
\ee
for $i=1,2$.
The error dynamics is rewritten as
\be
	\dot{e}=A_{cl}e+Bw
\ee
Augment the Lyapunov function $V$ with a quadratic term in $w$
\be
	V_a=V+\frac{1}{2}w^2.
\ee
Taking its derivative yields
\be
	\dot{V}_a=-\frac{1}{2}e^TQe+w\kappa
\ee
where
\be\label{kappa}
	\kappa=\Phi^\prime_{S2}(x_2+u_2)-\Phi^\prime_{S1}(-x_2+u_1)+m\dot{\Delta}_\Phi-m\ddot{r}+B^TPe.
\ee
Here, $\kappa$ is chosen such that $\kappa=-\gamma w$ with $\gamma >0$ to make $\dot{V}_a$ negative definite. Hence, the augmented system of $e$ and $w$ is asymptotically stable. It should be noted that we cannot deduce unique solutions of $u_1$ and $u_2$ from $\kappa$ in (\ref{kappa}).

From (\ref{kappa}),
\be\label{beta}
	\Phi^\prime_{S2}u_2-\Phi^\prime_{S1}u_1=\beta
\ee
where
\be
\nn \beta=-K_1\zeta-K_2 e-(\Phi^\prime_{S2}+\Phi^\prime_{S1})x_2+m\ddot{r}
\ee
with $K_1=\gamma+KB$, and $K_2=(KA+\gamma K+B^TP)e$.
The control redundancy can be resolved using the least square solution to (\ref{beta}), which solves the minimization of $u_1^2+u_2^2$. This minimization should indirectly reduce muscle activation inputs as virtual controls are muscle contraction velocities. Similar to \cite{RichterWarnerIFAC17}, the least square virtual control inputs are given as
\bea
	u_1&=&-\frac{\Phi^\prime_{S1}}{\Delta}\beta\\
    u_2&=&\frac{\Phi^\prime_{S2}}{\Delta}\beta
\eea
where $\Delta=(\Phi^\prime_{S1})^2+(\Phi^\prime_{S2})^2$.

\begin{remark}
Since nonlinear functions $\Phi_{Sj}$, $\Phi_{Pj}$, $f_j$, and $g_{j}^{-1}$ are defined on finite intervals and there are singularities, constrained techniques must be used to prevent a finite escape time.
\end{remark}

\begin{figure}[!t]
\centering
\includegraphics[width = \columnwidth]{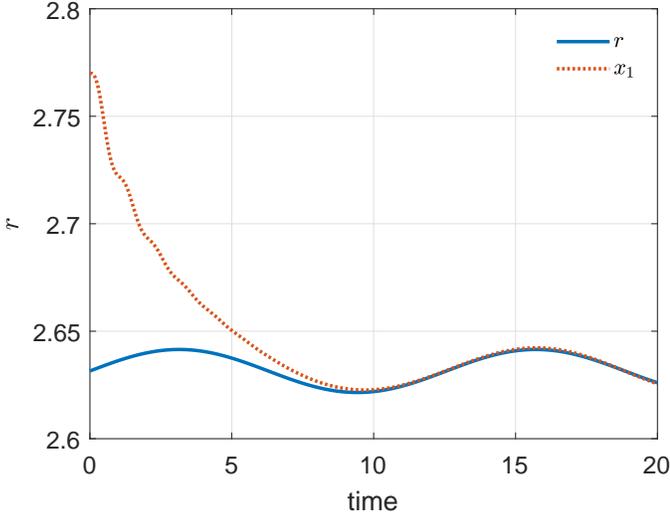}
\caption{The reference signal $r$ and the output $y_1=x_1$. All quantities are dimensionless (no units).} \label{r}
\end{figure}

\begin{figure}[!t]
\includegraphics[width = .7\columnwidth]{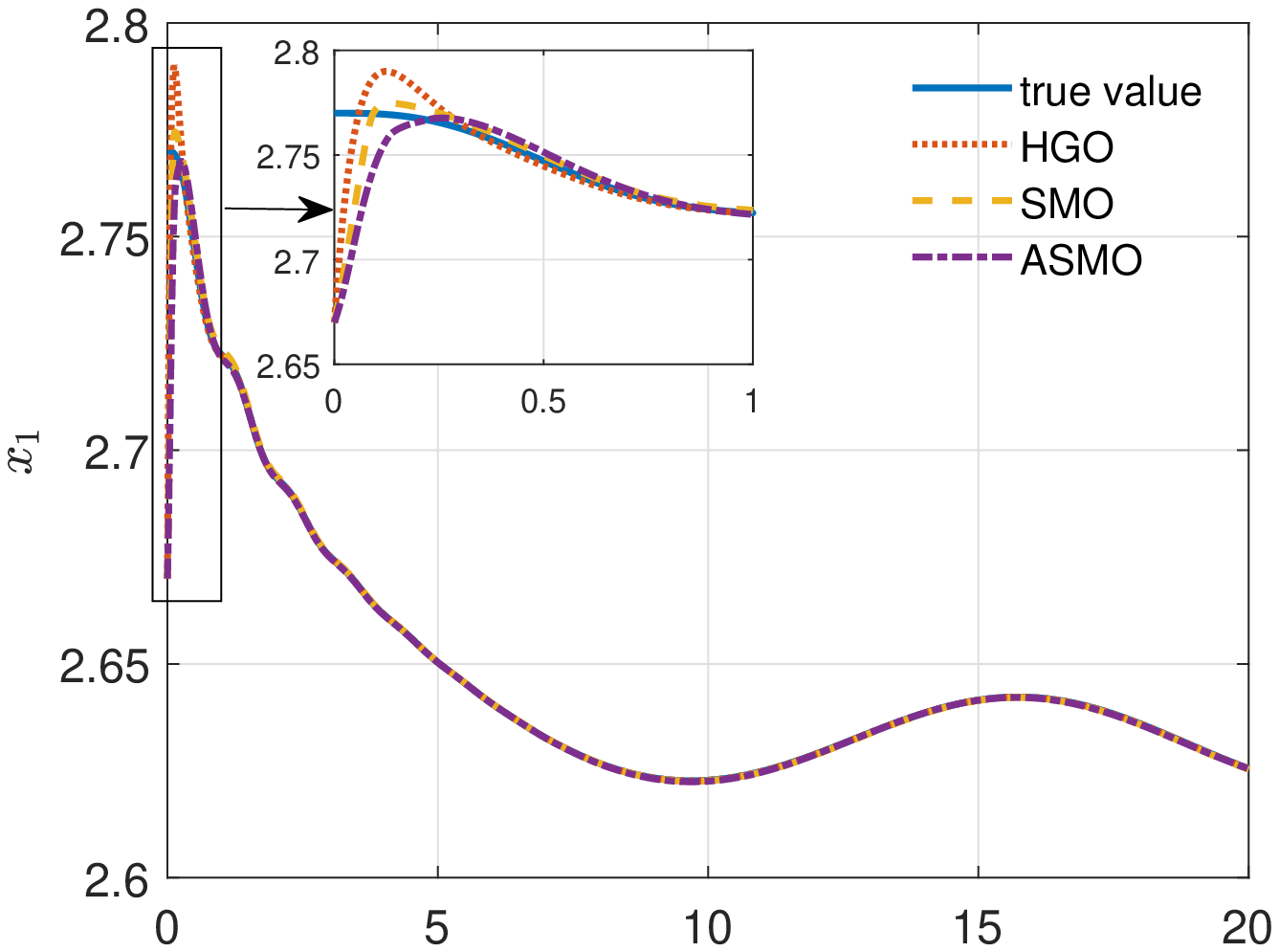}
\includegraphics[width = .7\columnwidth]{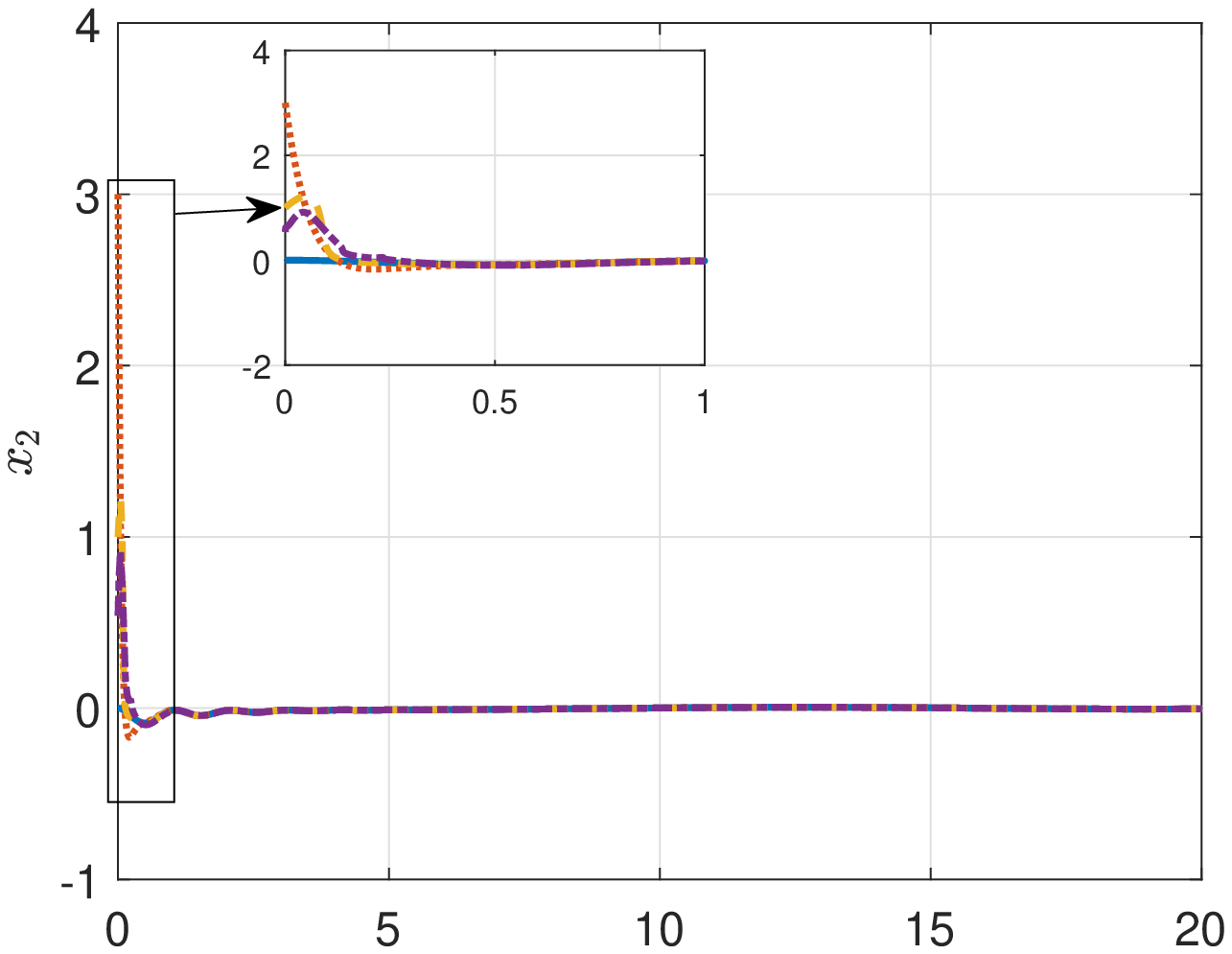}
\includegraphics[width = .7\columnwidth]{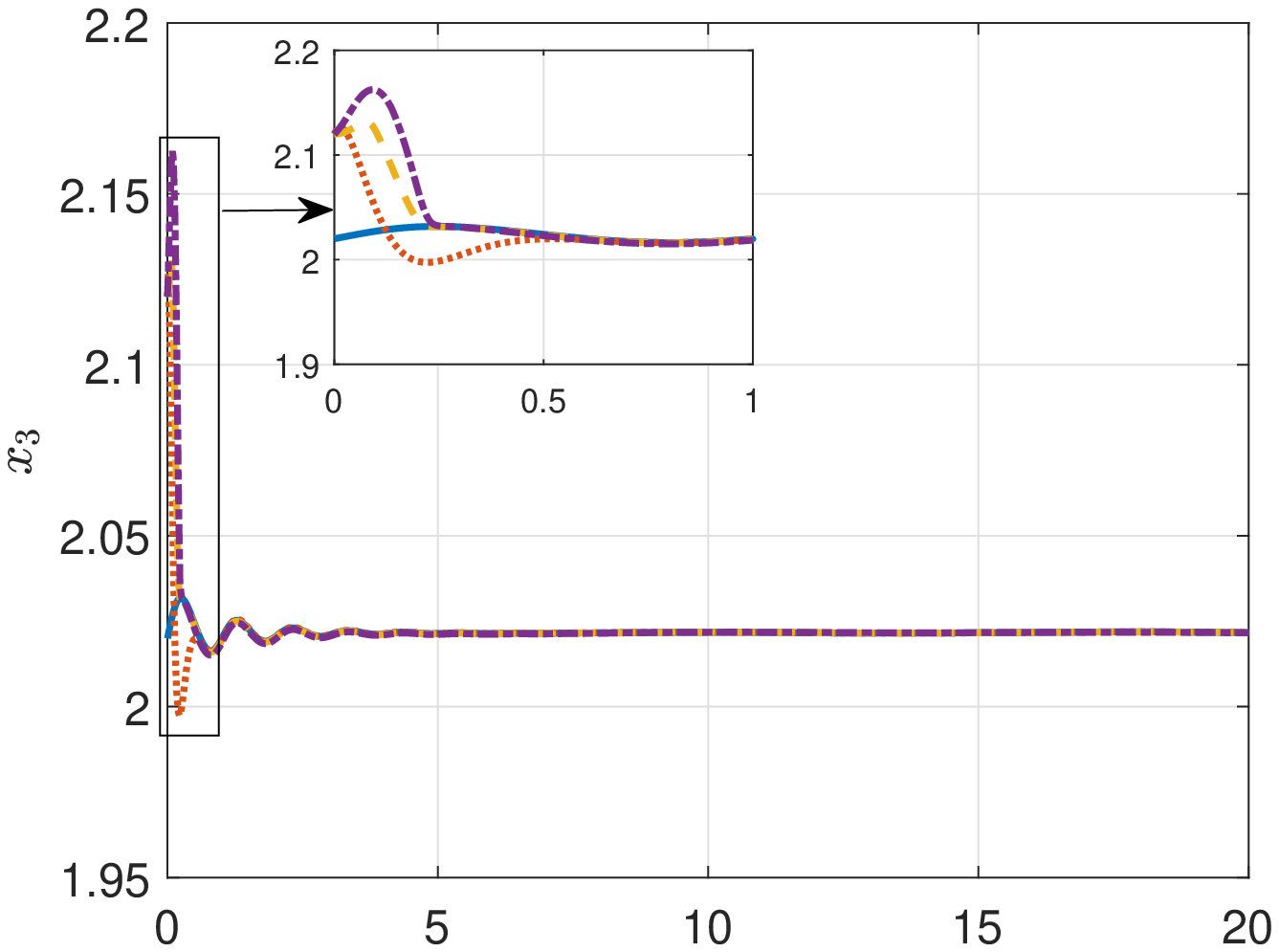}
\includegraphics[width = .7\columnwidth]{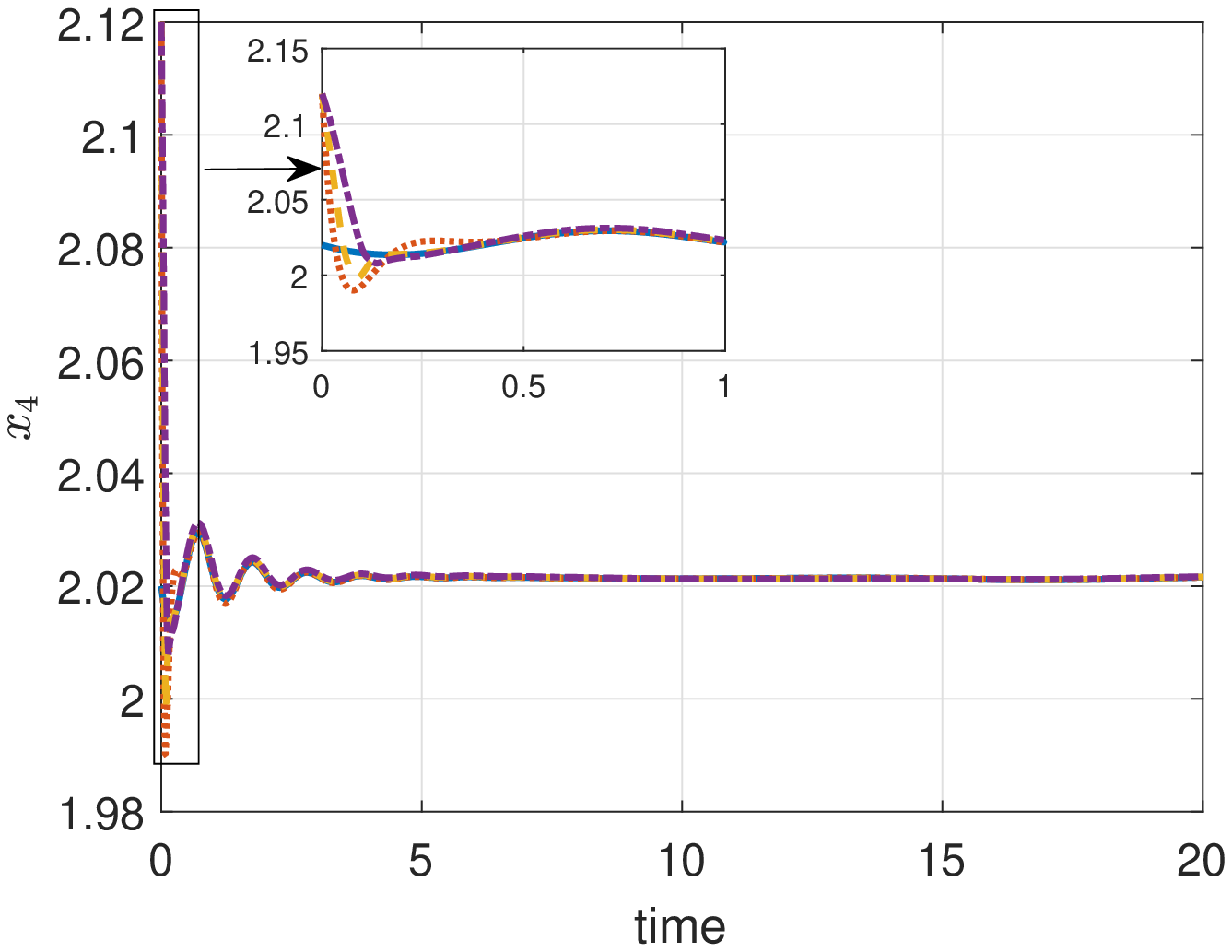}
\caption{The true value and estimates of $x$ for the noise free case using 3 observers: high gain observer (HGO), sliding mode observer (SMO), adaptive sliding mode observer (ASMO). All quantities are dimensionless (no units).} \label{x}
\end{figure}

\begin{figure}[!t]
\includegraphics[width = .7\columnwidth]{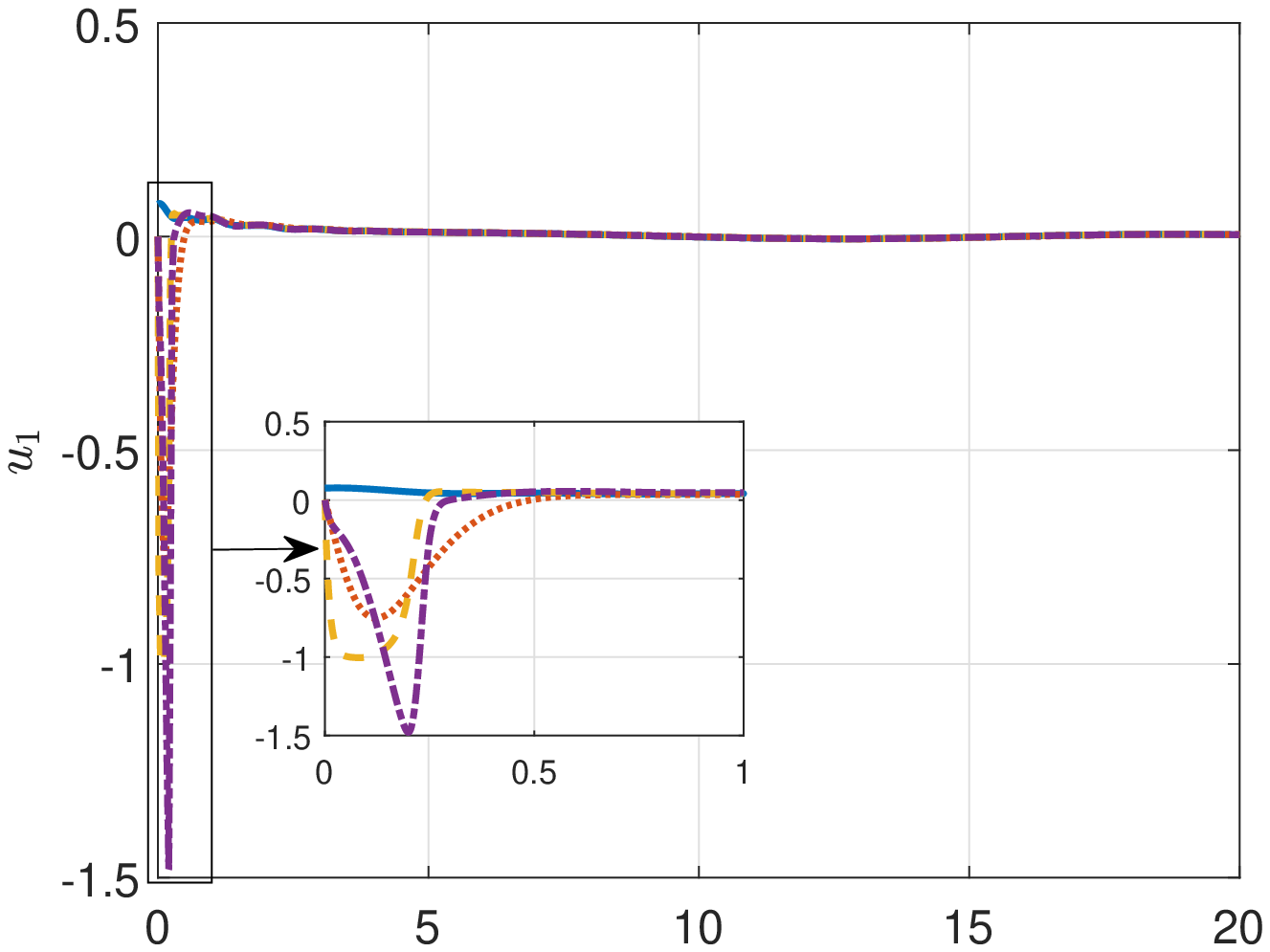}
\includegraphics[width = .7\columnwidth]{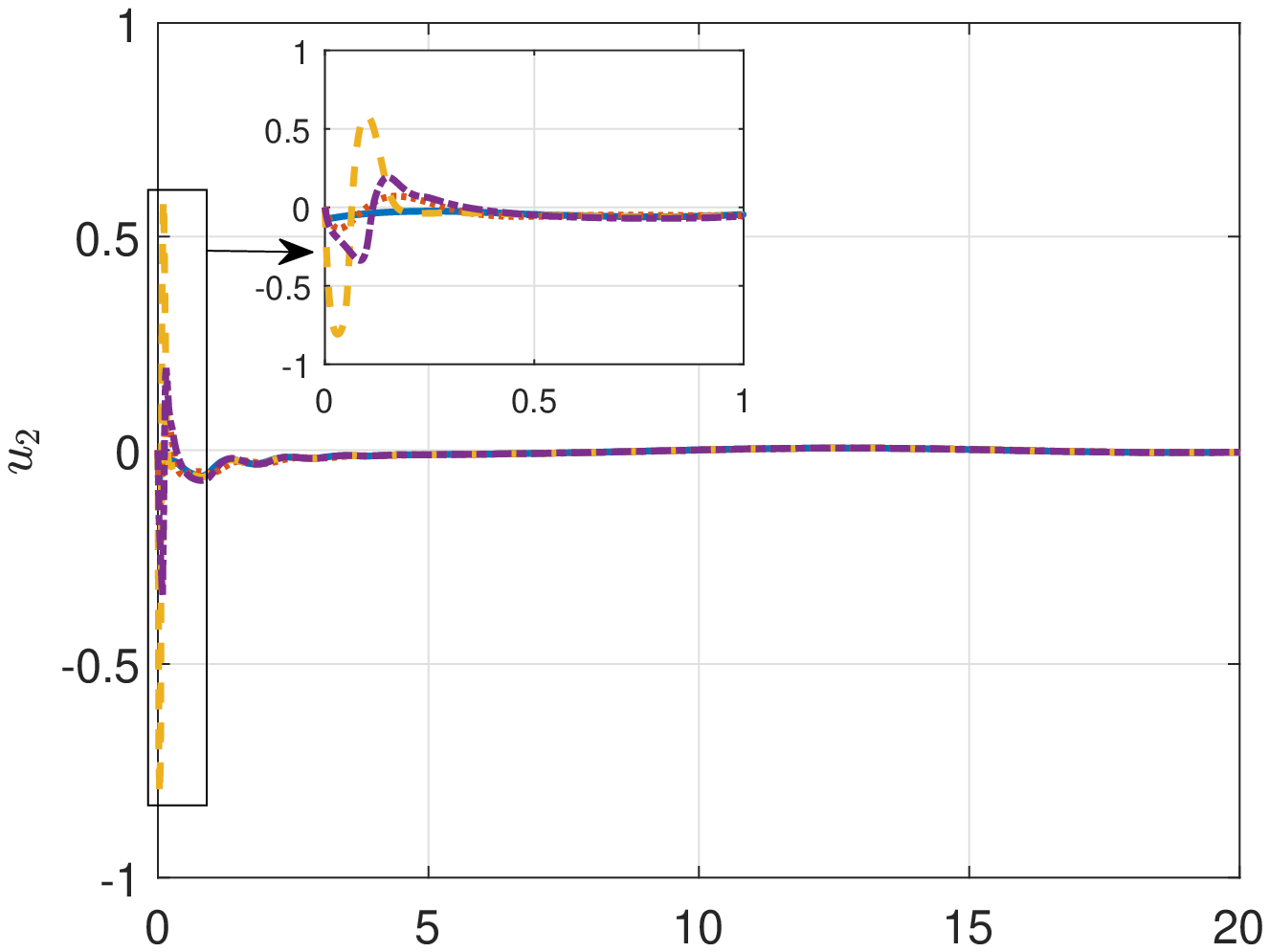}
\includegraphics[width = .7\columnwidth]{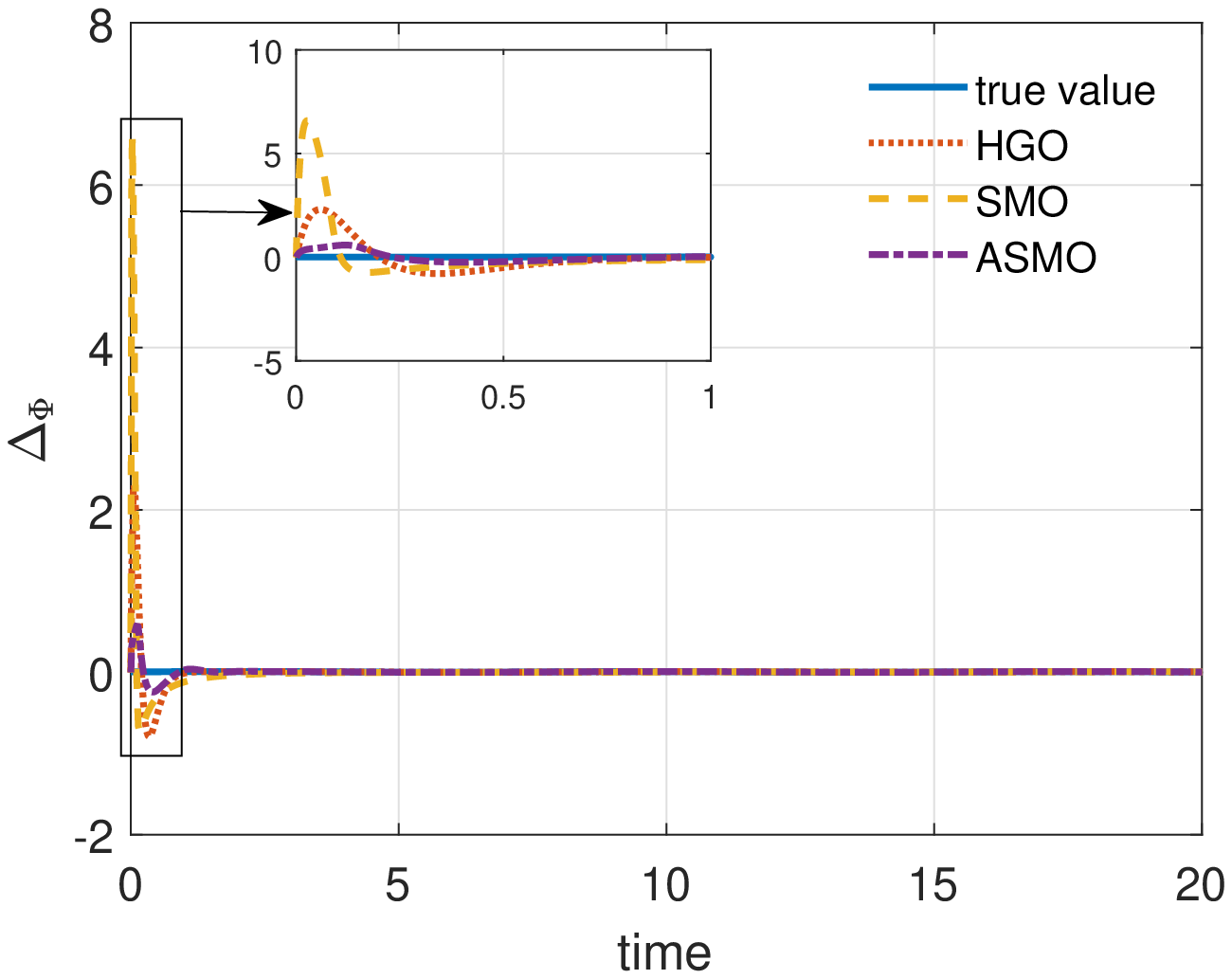}
\caption{The true value and estimates of $u$ and $\Delta_\Phi$ for the noise free case using 3 observers: high gain observer (HGO), sliding mode observer (SMO), adaptive sliding mode observer (ASMO). All quantities are dimensionless (no units).} \label{u}
\end{figure}

\begin{figure}[!t]
\includegraphics[width = .7\columnwidth]{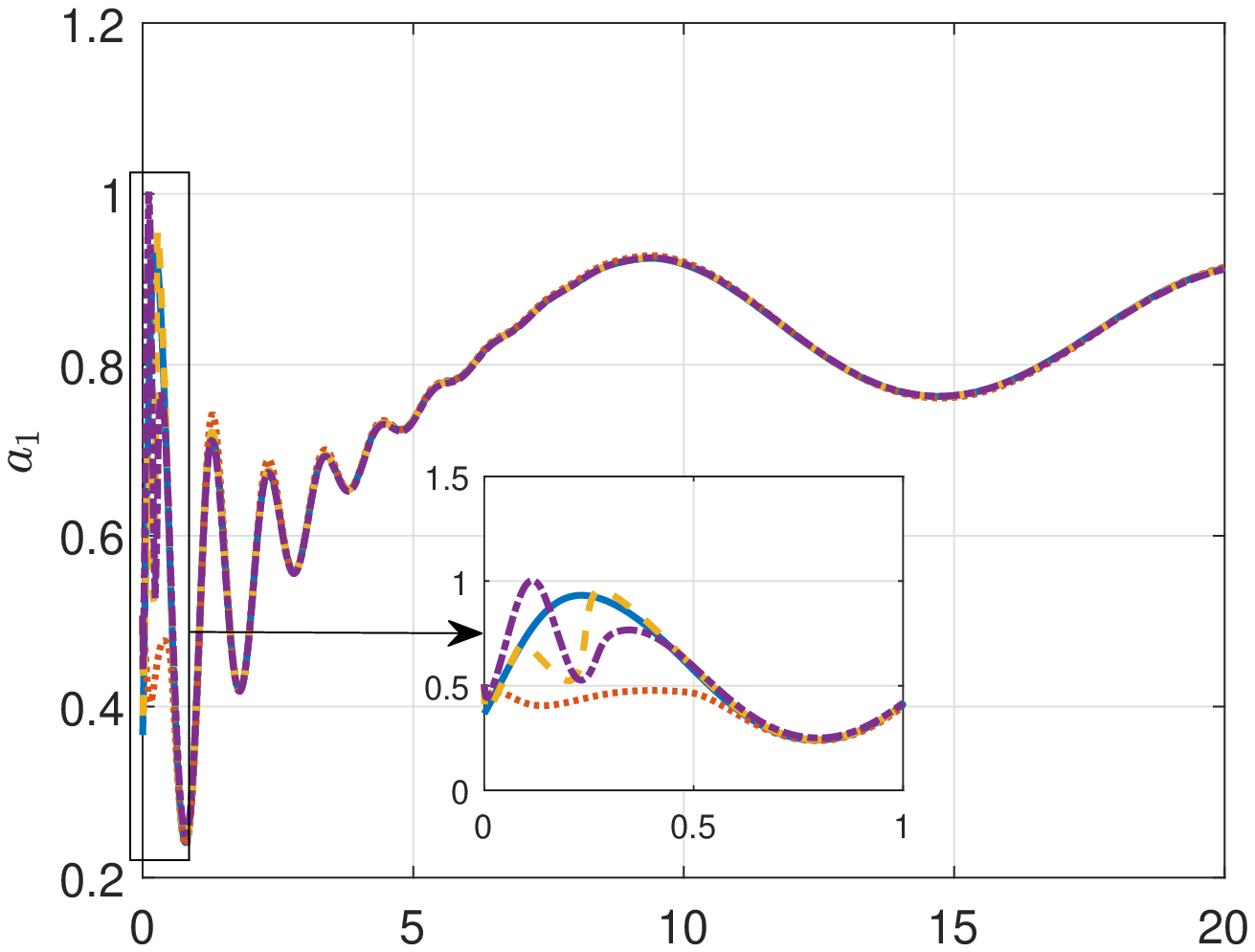}
\includegraphics[width = .7\columnwidth]{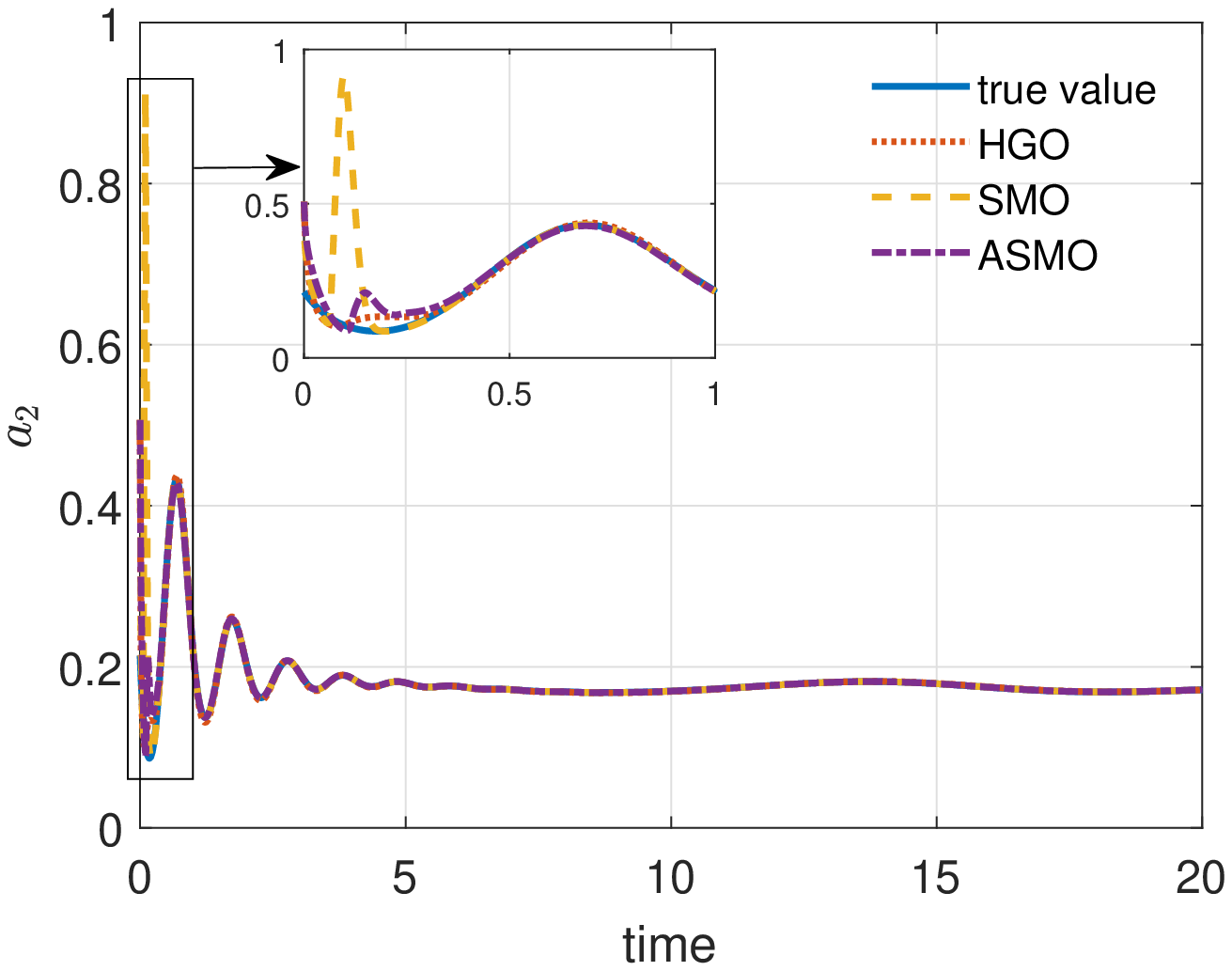}
\caption{The true value and estimates of $a$ for the noisy case using 3 observers: high gain observer (HGO), sliding mode observer (SMO), adaptive sliding mode observer (ASMO). All quantities are dimensionless (no units).} \label{a}
\end{figure}

\begin{figure}[!t]
\includegraphics[width = .7\columnwidth]{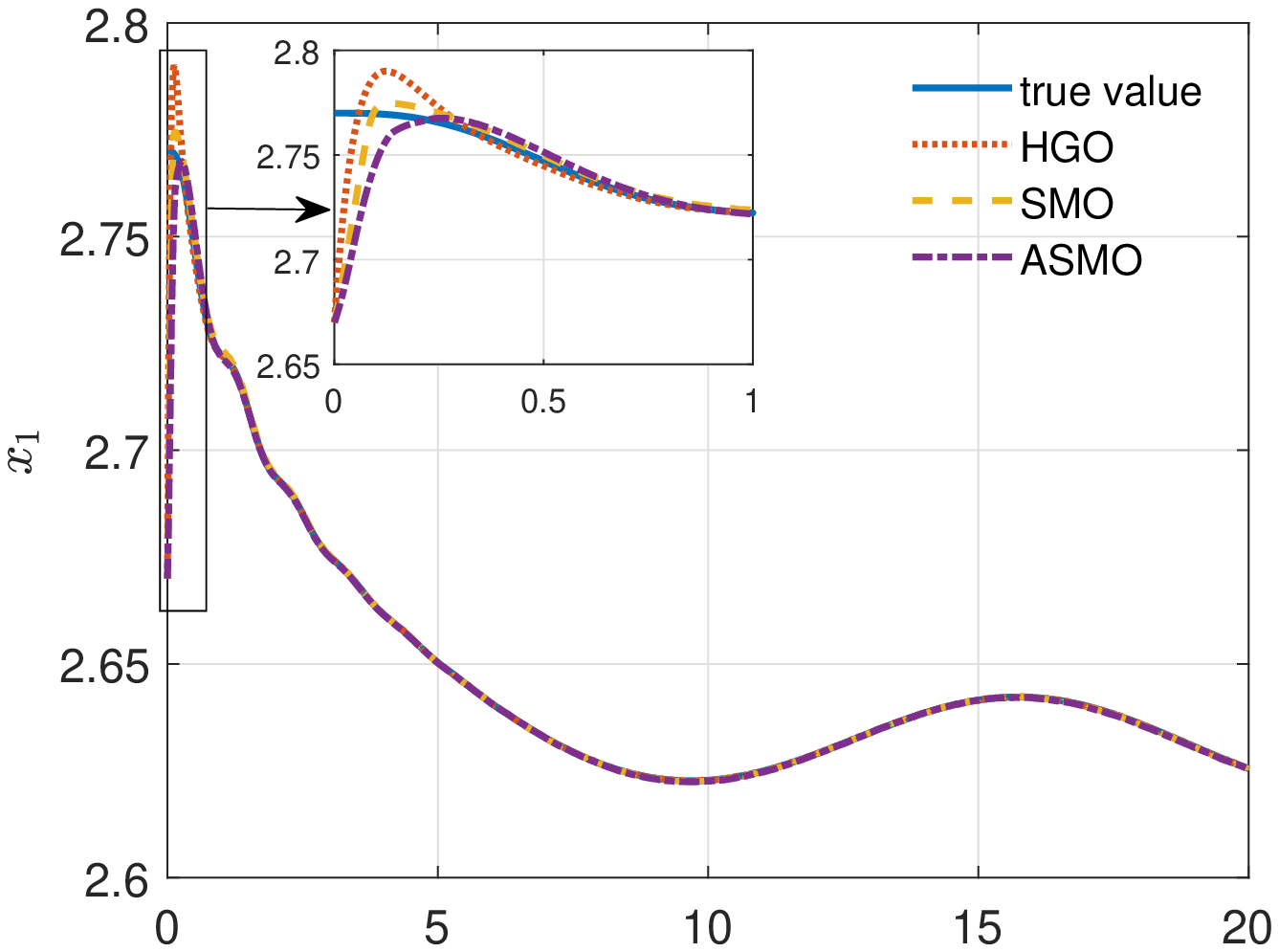}
\includegraphics[width = .7\columnwidth]{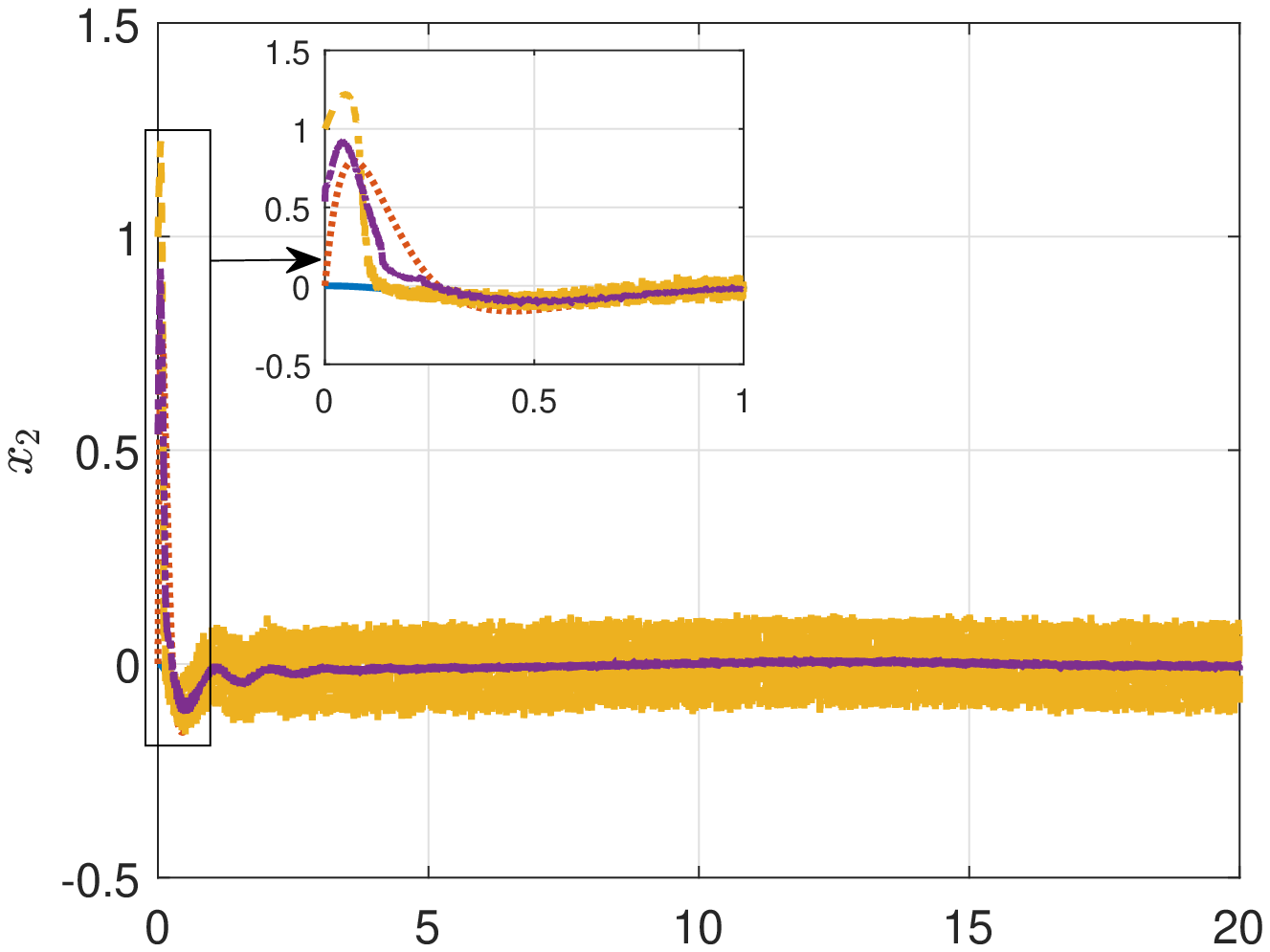}
\includegraphics[width = .7\columnwidth]{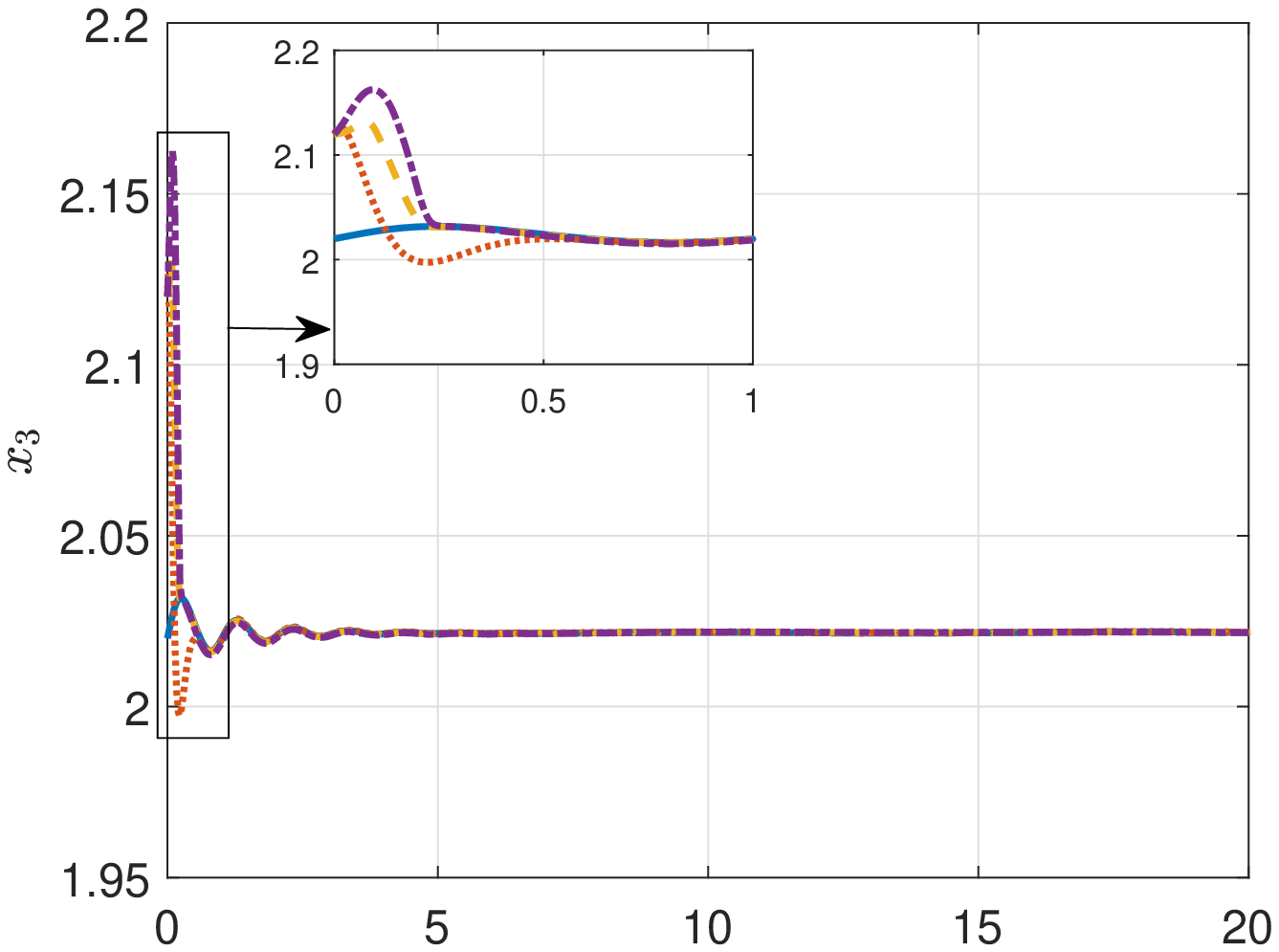}
\includegraphics[width = .7\columnwidth]{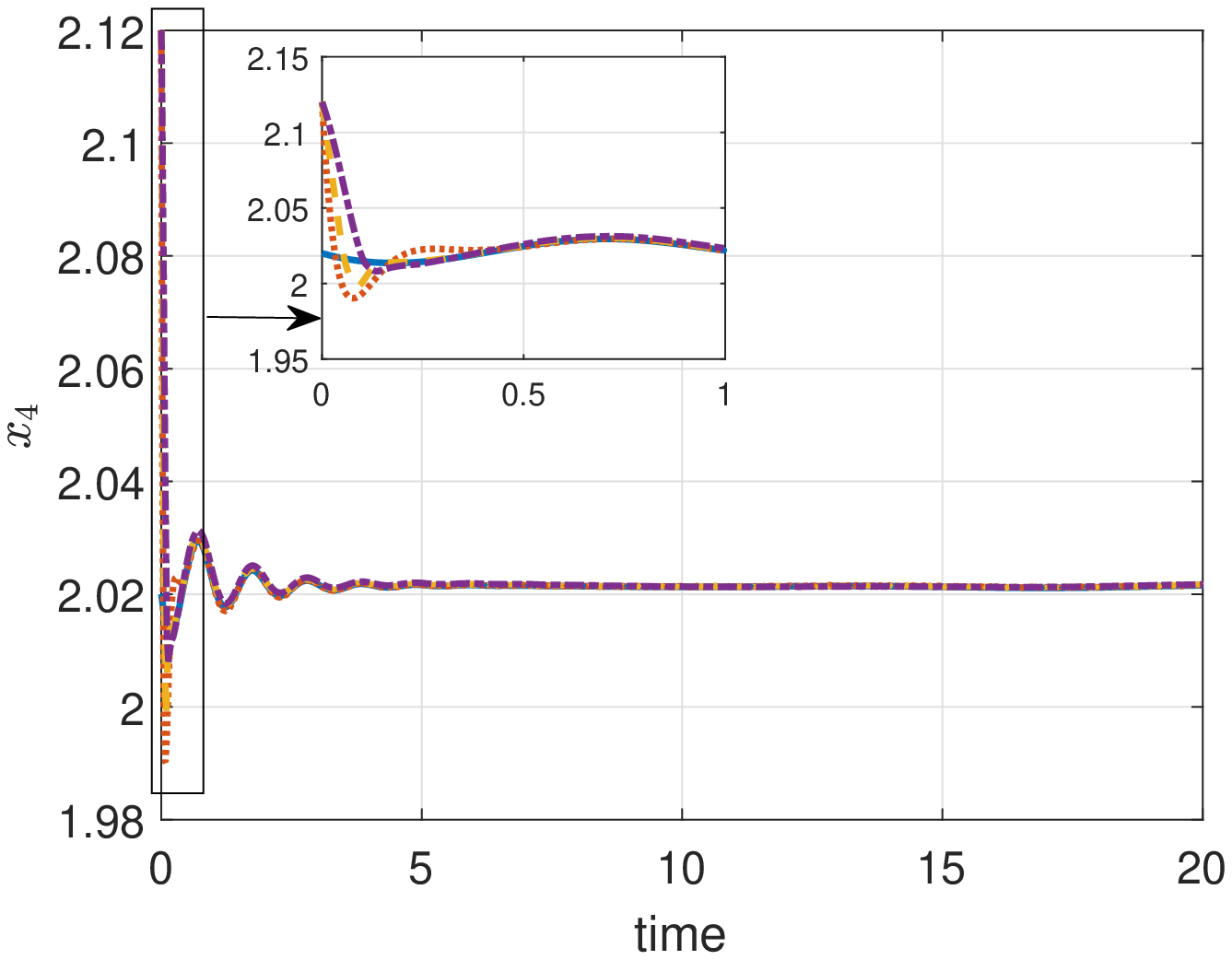}
\caption{The true value and estimates of $x$ for the noise free case using 3 observers: high gain observer (HGO), sliding mode observer (SMO), adaptive sliding mode observer (ASMO). All quantities are dimensionless (no units).} \label{xn}
\end{figure}

\begin{figure}[!t]
\includegraphics[width = .7\columnwidth]{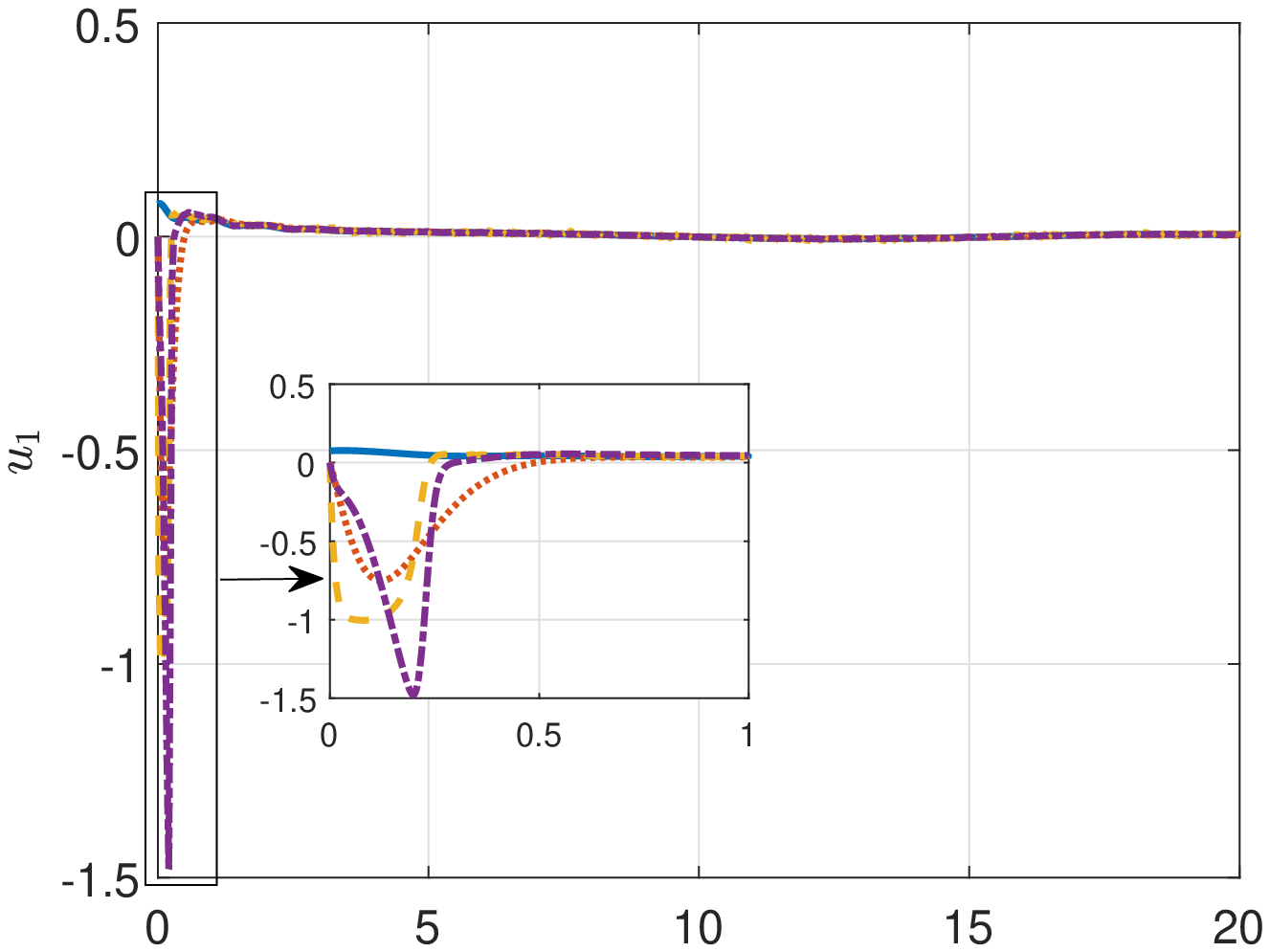}
\includegraphics[width = .7\columnwidth]{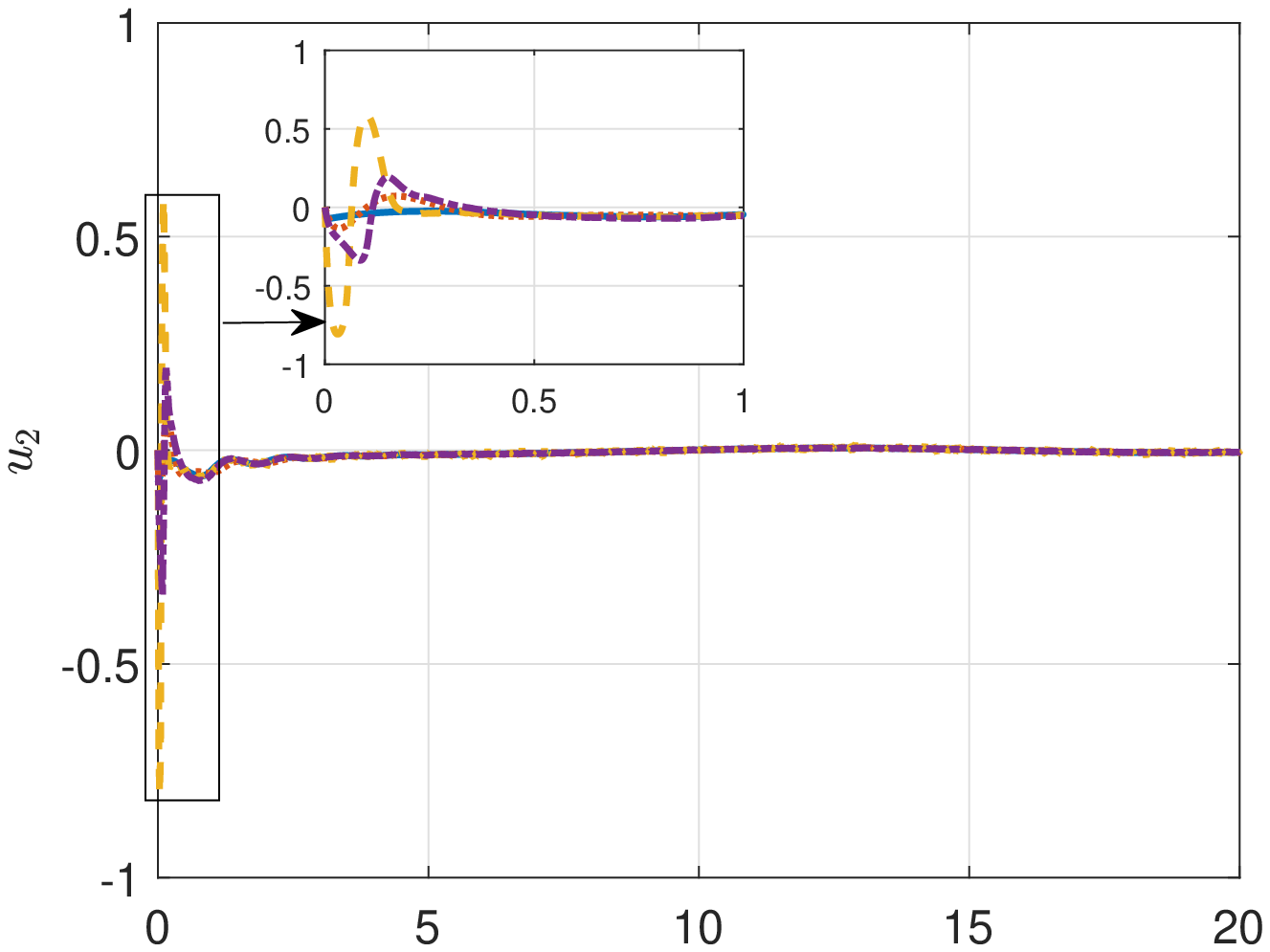}
\includegraphics[width = .7\columnwidth]{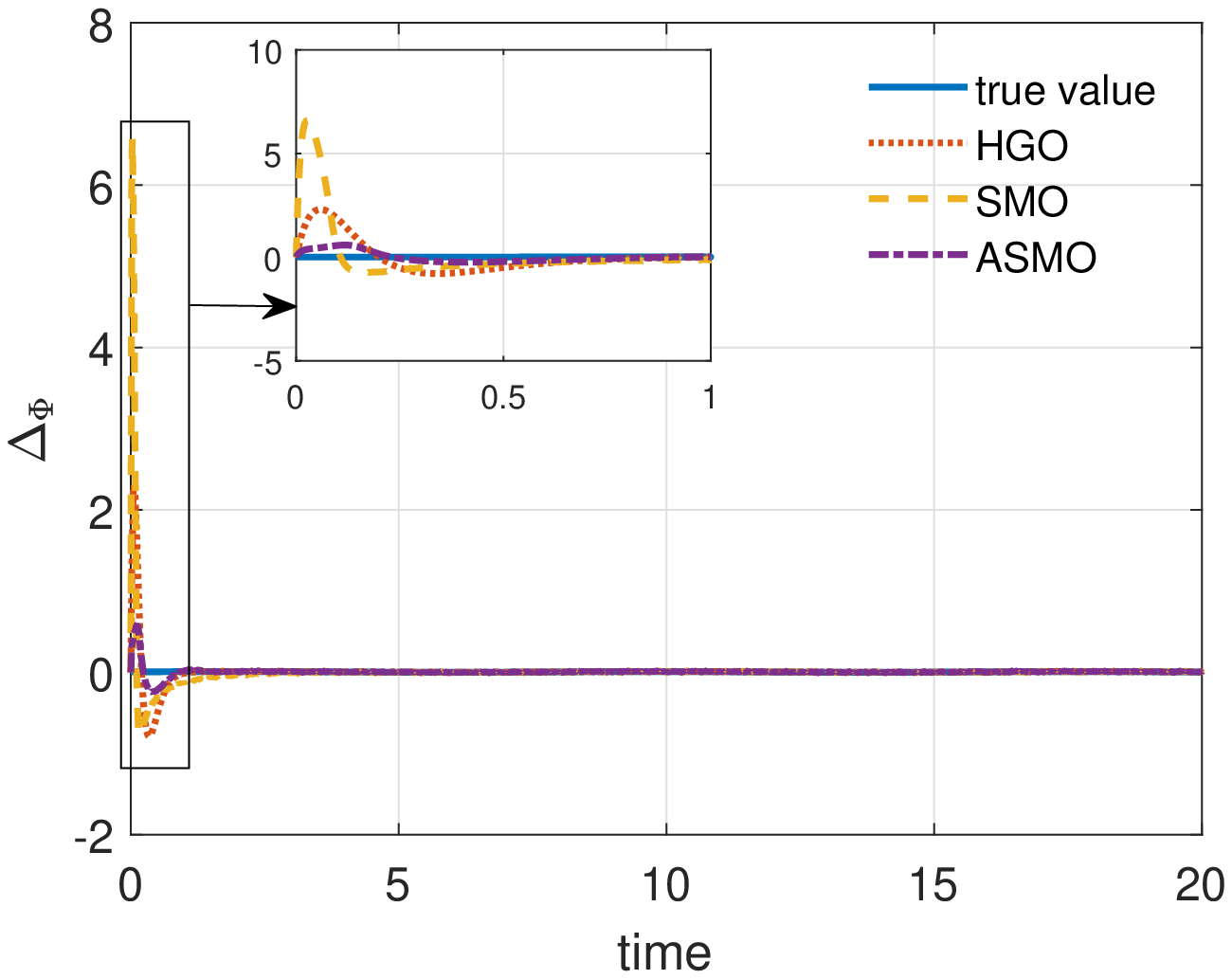}
\caption{The true value and estimates of $u$ and $\Delta_\Phi$ for the noisy case using 3 observers: high gain observer (HGO), sliding mode observer (SMO), adaptive sliding mode observer (ASMO). All quantities are dimensionless (no units).} \label{un}
\end{figure}

\begin{figure}[thp]
\includegraphics[width = .7\columnwidth]{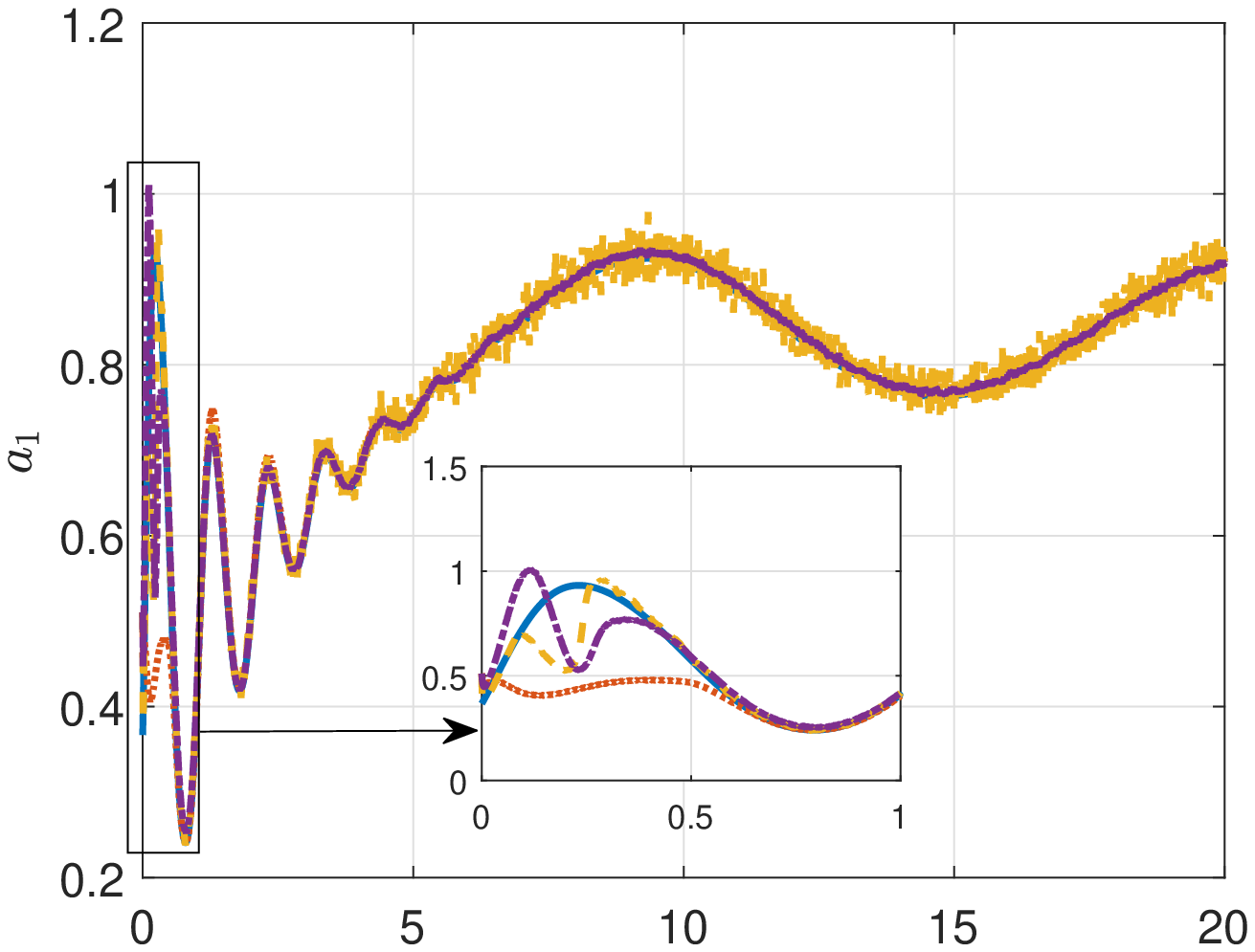}
\includegraphics[width = .7\columnwidth]{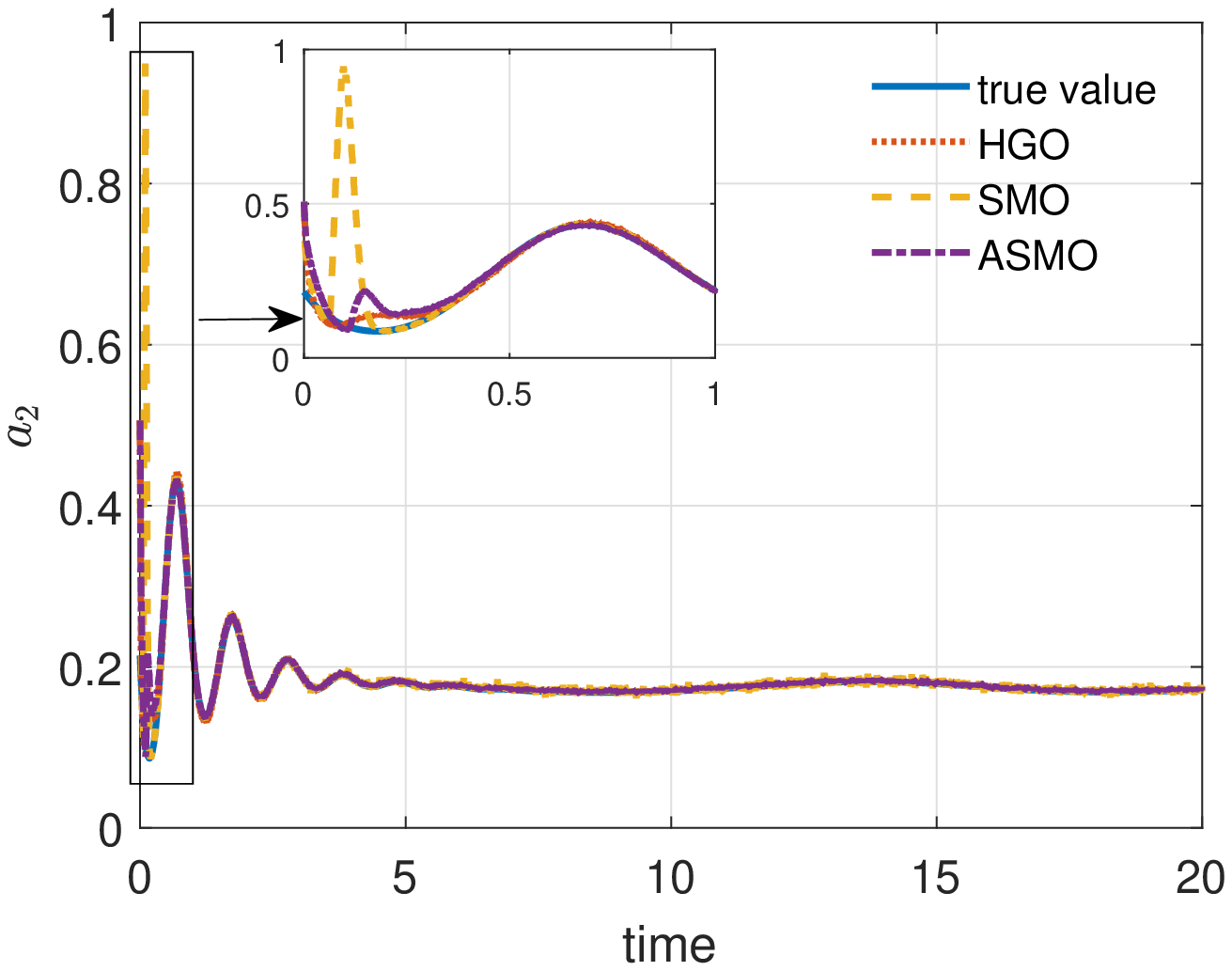}
\caption{The true value and estimates of $a$ for the noisy case using 3 observers: high gain observer (HGO), sliding mode observer (SMO), adaptive sliding mode observer (ASMO). All quantities are dimensionless (no units).} \label{an}
\end{figure}
\subsection{SIMULATION}\label{Simulation}

To illustrate the proposed scheme, we conducted two numerical simulations for a dual muscle system: noise free and noisy cases. The total length of the dual muscle system is $C=Lm_1+Lm_2=5.54$. The mass of the system is $m=1$.  The reference trajectory is chosen as
\be
\nn r=2.6315+0.01\,\sin 0.5 \tau.
\ee
Functions $\Phi_{Sj}$, $\Phi_{Pj}$  are chosen as in (\ref{PhiSj}) and (\ref{PhiPj}), \cite{WarnerRichterTR2016}. The parameter of (\ref{fj}) is $W=0.3$. The parameters of (\ref{gj}) are chosen as: $A=0.25$, $g_{max}=1.5$. Due to (\ref{PhiSj}), the upper bound of $\Phi_{S2}(x_4)-\Phi_{S1}(x_3)$ is 1.

The uncertainty of the system is
\be\label{DeltaPhi}
	\Delta_\Phi(\tau)=0.005+0.005\,\sin 0.8\tau.
\ee

The controller parameters in Subsection \ref{Control} are: $K=\left[\bmx 0.5774&1.2198\emx\right]$, $Q=\left[\bmx 10&0\\0&10\emx\right]$, $P=\left[\bmx 21.1284&17.3205\\17.3205&36.5955\emx\right]$, $\gamma=1$.

The parameters for the high gain observer presented in Section \ref{HGO} are: $\epsilon_h=0.1$, $h_{11}=3$, $h_{12}=3$, $h_{13}=1$, $h_{21}=2$, $h_{22}=1$, $h_{31}=2$, $h_{32}=1$. As pointed out in Section \ref{HGO}, $h_{11}$, $h_{12}$, $h_{13}$ are chosen such that the polynomial $s^3+h_{11}s^2+h_{12}s+h_{13}$ is Hurwitz, and $h_{ij}$ for $i=2,3$ and $j=1,2$ are chosen such that the polynomials $s^2+h_{i1}s+h_{i2}$ are Hurwitz for $i=2,3$. As the parameter $\epsilon_h$ is small, the convergence speed increases but when there is measurement noise, the performance of the observer is degraded \cite{Prasov_Khalil_TAC2013,AHRENS_Automatica2009}. Hence, $\epsilon_h$ should not be too small.

The parameters for the sliding mode observer presented in Section \ref{SMO} are:
$\alpha_{11}=1.1$, $\lambda_{11}=28.17$, $\alpha_2=1.1$, $\alpha_3=1$, $\tau_s=0.01$, $\delta_s=0.01$. The tuning of the parameters was shown in Section \ref{SMO}, in which $\alpha_{11}$ and $\lambda_{11}$ are chosen from (\ref{alpha11}), (\ref{lambda11}) where $p=0.5$ and $f^+=1$; $\alpha_2$ and $\alpha_3$ are chosen from (\ref{alpha2}) and (\ref{alpha3}). As $\delta_s$ converges to 0, the approximation (\ref{deltas}) becomes the ideal function $\mathrm{sign}$, which leads to high sensitivity to measurement noise. Hence, $\delta_s$ should not be too small to avoid degradation of the observer.

The parameters for the adaptive sliding mode observer presented in Section \ref{ASMO} are:
$\beta_0=1.1$, $\alpha_0=2\sqrt{2\beta_0}=2.97$, $\eta_1=0.2$, $\eta_2=0.2$, $a=0.82$, $l_0=0.4$, $r_{00}=0.4$, $r_{01}=0.5$, $r_{02}=0.5$, $\tau_a=0.01$, $\epsilon_{a1}=0.2$, $\epsilon_{a2}=0.2$, $\alpha_{a1}=0.99$, $\alpha_{a2}=0.99$, $\gamma_{a0}=200$, $\gamma_{a1}=\gamma_{a2}=300$, $\delta_{00}=\delta_{01}=\delta_{02}=0.001$, $\delta_a=0.01$. The parameters of $\beta_0$ and $\alpha_0$ are chosen according to \cite{Edwards_IJC2016} where $\alpha_0=2\sqrt{2\beta_0}$; $eta_1$ and $\eta_2$ in (\ref{AO2a}) and (\ref{AO3a}) are chosen as small numbers \cite{EDWARDS_Automatica2016}; $a$ in (\ref{deltaa0}) is chosen such as $0<a<1/\beta_0<1$  \cite{Edwards_IJC2016}; $l_0$ in (\ref{La}) and $r_{00}$ in (\ref{rhoa0}) are chosen as small positive values \cite{Edwards_IJC2016}; $r_{01}$ and $r_{02}$ in (\ref{rhoaj}) are small positive parameters \cite{EDWARDS_Automatica2016}; $\tau_a$ in lowpass filters (\ref{AO1c}), (\ref{AO2b}), (\ref{AO3b}) are chosen to be small; $\epsilon_{aj}$ and $\alpha_{aj}$ ($j=1,2$) are chosen such that $0<\alpha_{aj}<1$ and $\epsilon_{aj}>0$ to satisfy (\ref{uj}); $\gamma_{a0}$ in (\ref{dra0}) and $\gamma_{aj}$ in (\ref{draj}) ($j=1,2$) are positive; $\delta_{00}$ in (\ref{dra0})  and $\delta_{0j}$ ($j=1,2$) in (\ref{draj}) are small positive numbers; $\delta_a$ of the approximation function of the sign function in (\ref{deltaa}) is a small positive number. Similar to the sliding mode observer above, if $\delta_a$ is too close to 0, the observer will become degraded as this parameter is sensitive to measurement noise.

Note that the model under consideration is dimensionless as pointed out in Section \ref{DynamicModel}. Hence, there are no units specified on axes in the following figures.

In the first simulation, no noise affects the measurements of the system output. In Fig. \ref{r}, due to the presence of the uncertainty $\Delta_\Phi(\tau)$, $x_1$ is only able to be close to the reference signal after $\tau=8$, which demonstrates that the tracking control law is effective in producing a good tracking performance. It is shown in Fig. \ref{x} that the estimates of $x_1$, $x_2$, $x_3$, $x_4$ using the three observers converge to their true value at about $\tau=0.5$.  The estimates using the high gain observer experience peaks during transients. Fig. \ref{u} depicts the evolution of the estimates of the uncertainty $\Delta_\Phi$ and unknown inputs $u_1$ and $u_2$, which track well their true values. The estimates of the activation signals shown in Fig. \ref{a} converge to their true values. The closeness of the estimates and their true values reveals that the estimation schemes are effective in estimating the state variables and activation signals.

Next, the second simulation was conducted when the measurements were influenced by noise.  The noise affecting the measurement  signal of $x_1$ is uniformly distributed in the interval $[-0.001,0.001]$ and sampling time $T_s=0.005$. The measurements of the forces $\Phi_{S1}(x_3)$ and $\Phi_{S2}(x_4)$ are influenced by a noise profile which is a sum of a drift term of 0.001 and values uniformly distributed in the interval $[-0.001,0.001]$ with sampling time $T_s=0.005$. The estimates of $x$  in Fig. \ref{xn} look quite close to their counterparts in the noise free case (Fig. \ref{x}). Similarly, under the influence of the uncertainty $\Delta_\Phi(\tau)$, $x_1$ is close to the reference signal after $\tau=8$. The effect of measurement noise is much clearer in the evolutions of the estimates of $x_2$ in Fig. \ref{xn}. Here the estimate of $x_2$ using the adaptive sliding mode observer is slightly better than the two other observers. In Fig. \ref{un}, the estimates of $\Delta_\Phi$, $u_1$, and $u_2$ look a bit worse than in the noise free case (Fig. \ref{u}). The evolutions of the estimates of the activation signals in Fig. \ref{an} track well the true signals. It is shown that the estimates using the adaptive sliding mode observer are closest to the true values. These simulations demonstrate that our proposed estimation schemes produce reliable estimates of the state variables and activation signals in the presence of noise.

The two simulations illustrate that the three observers are comparably effective in estimating the state variables and activation signals of the dual muscle system. Note that the three observers have a lot of freedom in tuning parameters. While the adaptive sliding mode observer does not require knowledge of the bounds of the unknown inputs and uncertainty, the sliding mode observer offers more simple tuning with fewer parameters.

\section{CONCLUSIONS}\label{Conclusions}
In this paper, we have presented the agonistic-antagonistic muscle system based on the Hill muscle model. Three estimation approaches have been introduced to estimate the state variables and activation signals. The high gain observer is constructed based on recent development of the high gain estimation approach \cite{LEE_Automatica2015,LEE_SCL2016}. The sliding mode observer is designed based on the super twisting algorithm and first-order sliding mode \cite{davilaTAC2005,edwards1998book}. The adaptive sliding mode observer is developed based on dual layer adaptive sliding mode schemes presented  in \cite{EDWARDS_Automatica2016,Edwards_IJC2016}. Two numerical simulations were conducted to demonstrate the efficiency of the proposed schemes.

The traditional sliding mode observer is the most simple of the three observers with the least number of parameters but it requires the knowledge of the bounds of the uncertainty and unknown inputs. In contrast, the adaptive sliding mode observer estimates the system in an adaptive way without knowing the information of the uncertainty and unknown inputs at the cost of complexity. The high gain observer provides a flexible approach to observing the system. It was shown that the three observers are comparable through theoretical analysis and simulation results.

Our future work will investigate the estimation problem of more complicated multi-muscle multi-joint systems. In addition, experimental tests will be carried out to validate the proposed estimation schemes.

\bibliographystyle{IEEEtran}
\bibliography{IEEEabrv,Muscle_system_estimation}

\end{document}